\documentclass[a4paper]{article}

\usepackage{color}

\usepackage[margin=40mm]{geometry}

\usepackage{amsmath,amsthm,amssymb,latexsym,eufrak}

\newtheorem{theorem}{Theorem}[section]
\newtheorem{proposition}[theorem]{Proposition}
\newtheorem{lemma}[theorem]{Lemma}

\newtheorem{remark}[theorem]{Remark}
\newtheorem{definition}[theorem]{Definition}
\newtheorem{example}[theorem]{Example}
\newtheorem{assumption}[theorem]{Assumption}

\newcommand{\VB}{V(B)}
\newcommand{\vb}{v(B)}
\newcommand{\hb}{\bar{h}(B)}

\begin{document}

\title{On the existence of personal equilibria\thanks{The first author gratefully acknowledges the support of Universit\'{e} Paris-Saclay Springboard PIA Excellences, ANR 21-EXES-0003.  The second author gratefully 
acknowledges the support of the National Research, Development and Innovation Office (NKFIH) through grants K 143529,
KKP 137490 
and also within the framework of the Thematic Excellence Program 2021 (National Research subprogramme 
``Artificial intelligence, large networks, data security: mathematical foundation and applications'').}}

\author{Laurence Carassus\thanks{
Universit\'{e} Paris-Saclay, Centrale-Sup\'{e}lec, Math\'{e}matiques et Informatique pour la Complexit\'{e} et les Syst\`{e}mes and CNRS FR-3487, 91190, Gif-sur-Yvette, France;  laurence.carassus@centralesupelec.fr} 
\and Mikl\'os R\'asonyi\thanks{HUN-REN Alfr\'ed R\'enyi Institute of Mathematics and E\"otv\"os Lor\'and University, Budapest, 
Hungary; rasonyi@renyi.hu}}

\date{\today}

\maketitle{}

\begin{abstract}
We consider an investor who, while maximizing his/her expected utility, also compares the outcome to a reference entity.
We recall the notion of personal equilibrium and show that, in a multistep, generically incomplete financial market model
such an equilibrium indeed exists, under appropriate technical assumptions.
\end{abstract}


\smallskip{}



\section{Introduction}

It was first suggested by \cite{markowitz} that utility for an economic agent should be
defined not on wealth itself but on gains and losses relative to some reference point (present wealth in \cite{markowitz}).
Prospect theory, introduced in \cite{kt}, is also based on comparison to a reference point. Becoming a cornerstone of 
behavioural economics, this theory led to further developments, involving probability distortions, see \cite{tk}.
 
The papers \cite{bell,loomes} treated models with ``disappointment aversion'' where the actual outcome of
an investment is compared to an expected outcome via a gain-loss function. In \cite{kr2006} the
outcome of the investment and the reference point are compared pointwise. 
See \cite{review} for a review of reference-based 
preferences.

Our starting point is the model of \cite{kr2006}, further investigated in
\cite{kr2007,kr2008,kr2009}. A significant novelty of these papers is that the authors define the notion of \emph{personal equilibrium}: 
investors should rationally choose an action that is optimal when played against a reference point that is 
an independent copy of \emph{itself}.

The first natural question is whether such equilibria are realizable at all.
In one-step models, personal equilibria have been characterized in \cite{paolo-andrea1}, under mild conditions.
Characterization has been given for complete markets, too, in \cite{paolo-andrea2}. 

Most markets of interest, however, are incomplete and have multiple time steps. 
It is thus a fundamental question whether personal equilibria exist in such market models, too. 
As far as we know, this problem has not been addressed elsewhere yet.

In order to demonstrate that the notion of
personal equilibrium has a bearing on practically relevant situations, we will prove its existence
in a fairly general setting of multi-step (generically incomplete) models. In this way, we provide a reassuring theoretical
guarantee that this notion is non-void for sufficiently complex models of financial markets.  
As usual, our equilibrium considerations will be based on the existence of fixed points which, in the present case, requires
rather involved arguments. We show continuous dependence of the strategies ``on the past'' which allows to apply Schauder's 
fixed point theorem in a Banach space of continuous functions. We further rely on results of
\cite{cr2007b} where continuity of strategies with respect to preferences was established.

Section \ref{sec2} rigorously formulates our assumptions and main theorem. Proofs are given in Section \ref{sec3}, while
Section \ref{sec4} presents certain technical results that are used in our main line of arguments.

\section{Model assumptions and results}\label{sec2}
Throughout this paper, we will be working on a probability space 
$(\Omega,\mathcal{F},\mathbb{P})$. All sigma-algebras will
be assumed to be completed with respect to $\mathbb{P}$, without further mention. 
Expectation under $\mathbb{P}$ will be denoted by $\mathbb{E}[\cdot]$, $\mathbb{R}_{+}:=\{x\in\mathbb{R}:x\geq 0\}$,  $|x|$ is for the Euclidean norm of $x\in \mathbb{R}^k$, whatever $k\geq1$, 
and $\mathcal{B}(X)$ designates the Borel sigma-algebra on any topological space $X$. 

\subsection{Hypotheses on the financial market model}

We first elaborate on the information structure.
We postulate that the filtration is generated by a sequence of bounded independent random variables, and the probability space is large enough to support an auxiliary random variable that will be used in the statements
of our results below.

\begin{assumption}\label{filtra} Let $m\geq 1$ be an integer. 

\begin{enumerate}

\item Let
$\varepsilon_{t}$, $1\leq t\leq T$ be $\mathbb{R}^{m}$-valued independent random variables. The investor's decisions 
are based on the (completed) natural filtration $\mathcal{F}^{\varepsilon}_{t}:=\sigma(\varepsilon_{1},
\ldots,\varepsilon_{t})$, $1\leq t\leq T$. 
($\mathcal{F}^{\varepsilon}_{0}$ coincides with the $\mathbb{P}$-null sets.)

\item Moreover, the $\varepsilon_{t}$ are bounded, say, $|\varepsilon_{t}|\leq C_{\varepsilon}$, $1\leq t\leq T$ for some
constant $C_{\varepsilon}$. 

\item There is a random variable $\hat{\varepsilon}$ which is 
independent of $\mathcal{F}^{\varepsilon}_{T}$ and is uniformly distributed on $[0,1]$. 

\end{enumerate}
\end{assumption}

An element of $\mathbb{R}^{T\times m}$ will be denoted most often by $e$, where $e=(e_{1},\ldots,e_{T})$ with $e_{t}\in\mathbb{R}^{m}$
for $1\leq t\leq T$. If $e\in\mathbb{R}^{T\times m}$ then $e^{t}$ will refer to $(e_{1},\ldots,e_{t})$, for $1\leq t\leq T$.

A risky asset with price $S_{t}$ at time $t$ will be considered.
We stipulate that
price increments are H\"older-continuous, bounded functions of the factors generating the investor's filtration $(\mathcal{F}^{\varepsilon}_{t})_{0\leq t\leq T}$.


\begin{assumption}\label{price1} The initial price $S_{0}$ is constant.
For $1\leq t\leq T$ there exist functions $f_{t}:\mathbb{R}^{t\times m}\to \mathbb{R}$ such that 
$$
\Delta S_{t}:=S_{t}-S_{t-1}=f_{t}(\varepsilon_{1},\ldots,\varepsilon_{t}).
$$
For all $1\leq t\leq T$ and for all $e^{t},\bar{e}^{t}\in\mathbb{R}^{t\times m}$,
\begin{eqnarray}
|f_{t}(e^{t})-f_{t}(\bar{e}^{t})|\leq C_{f}|e^{t}-\bar{e}^{t}|^{\chi},
\label{lipf}
\end{eqnarray}	
for some $C_{f}>0$, $0< \chi \leq 1,$ and for all $1\leq t\leq T$,
\begin{eqnarray}
\sup_{e^{t}\in \mathbb{R}^{t\times m}}|f_{t}(e^{t})|\leq C_{f}.
\label{unifbornef}
\end{eqnarray}	
\end{assumption}
In particular, Assumption \ref{price1} implies that each $f_{t}$ is Borel measurable and the process $(S_{t})_{0\leq t\leq T}$ is adapted to 
the filtration $(\mathcal{F}_{t}^{\varepsilon})_{0\leq t\leq T}$.
\begin{remark}\label{compact} {\rm It is enough to postulate \eqref{lipf} and \eqref{unifbornef} on the compact set $K_t:=[-C_{\varepsilon},C_{\varepsilon}]^t$ for all $1\leq t\leq T$, see Proposition \ref{lipcomp} below.
}
\end{remark}

The next ``uniform no-arbitrage'' assumption has already been used multiple times in optimal investment problems,
see \cite{cr2007a,cr2007b}.
It expresses that future price movements conditioned to the past make a move of at least a prescribed size both up and down, with at least a fixed positive probability. The word ``uniform'' comes from the fact that
$\alpha$ does not depend on $e^{t-1}$.

\begin{assumption}\label{una} There is $0<\alpha\leq 1$ such that, for all $1 \leq t \leq T$, for all $e^{t-1}\in\mathbb{R}^{(t-1)\times m}$,
\begin{equation}\label{ep}
\mathbb{P}[f_{t}(e^{t-1},\varepsilon_{t})\geq \alpha]\geq \alpha,\quad 
\mathbb{P}[f_{t}(e^{t-1},\varepsilon_{t})\leq -\alpha]\geq \alpha.
\end{equation}
(In the case $t=1$ we mean $f_{1}(e^{0},\varepsilon_{1}):=f_{1}(\varepsilon_{1})$.)
\end{assumption}

\begin{remark}\label{conpr} {\rm Translated into the language of conditional probabilities, \eqref{ep} means} 
$$
\mathbb{P}[\Delta S_{t}\geq \alpha\vert\mathcal{F}^{\varepsilon}_{t-1}]\geq \alpha,\quad 
\mathbb{P}[\Delta S_{t}\leq -\alpha\vert\mathcal{F}^{\varepsilon}_{t-1}]\geq \alpha  \quad \mbox{a.s.}
$$
\end{remark}

\begin{example}\label{eex}{\rm  
Let for all $1 \leq t \leq T$, $\mu_{t},\sigma_{t}:\mathbb{R}^{t-1}\to\mathbb{R}$
be bounded H\"older continuous functions (we mean that $\mu_{1},\sigma_{1}$ are constants). 
Let $0<\delta\leq 1$, $C>0$ be such that $|\sigma_{t+1}|+|\mu_{t+1}|\leq C$ and 
$$
|\mu_{t+1}(e^{t})-\mu_{t+1}(\bar{e}^{t})|+|\sigma_{t+1}(e^{t})-\sigma_{t+1}(\bar{e}^{t})|\leq C|e^{t}-\bar{e}^{t}|^{\delta},
$$ 
for all $0 \leq t \leq T-1$ and for all $e^{t},\bar{e}^{t}\in\mathbb{R}^{t}$. We assume that $\sigma_t\geq c>0$. 
Let $\varepsilon_{t}$, $1 \leq t \leq T$ be $\mathbb{R}$-valued  bounded independent random variables (that is, $m=1$) such that there exists 
$\beta>0$ with 
$\mathbb{P}\bigl[\varepsilon_{t}\leq\frac{-C-\beta}{c}\bigr]\geq\beta$ and 
$\mathbb{P}\bigl[\varepsilon_{t}\geq \frac{C+\beta}{c}\bigr]\geq\beta$ for all $1 \leq t \leq T$.  
Let $S_{0}$ be constant and define the price process for all $1 \leq t \leq T-1$, recursively as
$$
\Delta S_{t+1}=\mu_{t+1}(\varepsilon_{1},\ldots,\varepsilon_{t})+\sigma_{t+1}(\varepsilon_{1},\ldots,\varepsilon_{t})\varepsilon_{t+1}.
$$
Then Assumptions \ref{price1} and \ref{una} hold, see Lemma 
\ref{vege} below.
}
\end{example}

Now we describe the investor's activities in the market.
Fix initial capital $x_{0}\in\mathbb{R}$. Predictable trading strategies
$\phi=(\phi_{1},\ldots,\phi_{T})$ are such that $\phi_{t}$ is a $\mathcal{F}^{\varepsilon}_{t-1}$-measurable
random variable, representing the investor's position in the risky asset at time $t$, for all $1\leq t\leq T$.
The set of all such strategies is denoted by $\Phi$. We assume that there is a bank account with $0$ interest rate in the market.

Trading according to $\phi\in\Phi$ in a self-financing way, portfolio value from initial capital $x_{0}$ at times $0 \leq t \leq T$ is then defined as 
$$
W_{t}(x_{0},\phi):=x_{0}+\sum_{j=1}^{t}\phi_{j}\Delta S_{j}.
$$

We remark that, by Doob's theorem, for each $\phi\in\Phi$ one can find Borel measurable 
functions $\varphi_{t}:\mathbb{R}^{(t-1)\times m}\to\mathbb{R}$,
$1 \leq t \leq T$ such that 
\begin{equation}\label{doob}
\phi_{t}=\varphi_{t}(\varepsilon_{1},\ldots,\varepsilon_{t-1}).	
\end{equation} (In the case $t=1$ we mean that $\phi_{1}=\varphi_{1}$ is
a constant). Keeping this in mind, we now introduce the reference points corresponding to possible portfolio strategies.

The ``extra randomness'' provided by $\hat{\varepsilon}$ in Assumption \ref{filtra} enables us to fabricate an
independent copy $(\hat{\varepsilon}_{1},\ldots,\hat{\varepsilon}_{T})$ of $(\varepsilon_{1},\ldots,\varepsilon_{T})$ in the next lemma,
whose proof is relegated to Section \ref{sec4}.

\begin{lemma}\label{epsilontilde}
There is a Borel measurable function $\Upsilon:[0,1]\to\mathbb{R}^{T\times m}$ such that $\Upsilon(\hat{\varepsilon})$ has the
same law as $(\varepsilon_{1},\ldots,\varepsilon_{T})$.	We shall write $(\hat{\varepsilon}_{1},\ldots,\hat{\varepsilon}_{T})$
for $\Upsilon(\hat{\varepsilon})$. 
\end{lemma}

For any strategy ${\phi}\in\Phi$, we define ``the independent copy of its final wealth'' by
\begin{eqnarray}
B({\phi}):=x_{0}+\sum_{t=1}^{T}\varphi_{t}(\hat{\varepsilon}_{1},\ldots,\hat{\varepsilon}_{t-1})f_{t}(\hat{\varepsilon}_{1},\ldots,\hat{\varepsilon}_{t}),
\label{bebe}
\end{eqnarray}
where the $\varphi_{t}$ are as in \eqref{doob}. We sometimes write $B({\phi})=W_{T}(x_{0},\phi)(\hat{\varepsilon}^{T}).$

This definition deserves some explanation. We imagine that independent copies 
$f_{t}(\hat{\varepsilon}_{1},\ldots,\hat{\varepsilon}_{t})$ of the price increments $\Delta S_{t}$
exist ``somewhere'' and one can trade in this
asset following the strategy $\phi$, but using the driving factor $\hat{\varepsilon}_{t}$ instead of $\varepsilon_{t}$. 
To cut a long story short, $B(\phi)$ is independent of the financial market the investor is trading in, but its distribution
is the same as that of $W_{T}(x_{0},\phi)$. In the model proposed by \cite{kr2006}, the investor compares his/her
portfolio performance to a reference entity of such type, see the next subsection for more details.

\subsection{Hypotheses on the investor's preferences}

We consider a utility function satisfying the following properties.

\begin{assumption}\label{utility} 

\begin{enumerate}
\item $U:\mathbb{R}\to\mathbb{R}$ is twice continuously differentiable.

\item $U$ is non-decreasing, bounded from above, that is, there is $C_{U}>0$ such that $U(x)\leq C_{U}$ for all $x\in\mathbb{R}$.

\item For all $x$, $U''(x)<0$.
\end{enumerate}	
\end{assumption}

\begin{remark}
{\rm Clearly, Assumption \ref{utility} implies that $U$ is strictly concave and strictly increasing. Also, $|U''|$ and $U'$ are bounded away from $0$
on every bounded set.}
\end{remark}

Next, we list our hypotheses about the gain-loss function $\nu$.

\begin{assumption}\label{mienk}

\begin{enumerate}
\item $\nu:\mathbb{R}\to\mathbb{R}$ is twice continuously differentiable, $\nu(0)=0$.

\item $\nu$ is bounded from above, and $\nu''$ is bounded.  That is, there is $C_{\nu}>0$ such that $\nu(x)\leq C_{\nu}$ and $|\nu''(x)|\leq C_{\nu}$ for all $x\in\mathbb{R}$.

\item $\nu(x)=k_{-} x$ for $x\leq 0$ with some $k_{-}>0$.

\item On $(0,\infty)$, $\nu'$ is nonincreasing and $0< \nu'<k_{-}$.
\end{enumerate}	
\end{assumption}

\begin{remark}{\rm Note that under 
Assumption \ref{mienk}, $\nu$ is concave and stricly increasing on $\mathbb{R}$, and, 
$$
0<\nu'(x)\leq k_{-},\ x\in\mathbb{R}.
$$
Using the Newton-Leibniz rule, this implies, for all $x,y \in \mathbb{R}$ 
\begin{equation}\label{nuleibnitz}
|\nu(x)-\nu(y)|=\big| \int_{y}^{x}\nu'(t)\, dt \big|\leq k_{-}|x -y|. 
\end{equation}
}
\end{remark}

\begin{remark}
{\rm In \cite{kr2006} the following assumptions were made about $\nu$:
\begin{itemize}
\item[A0] $\nu$ is continuous, twice continuously differentiable on $\mathbb{R}\setminus \{0\}$, $\nu(0)=0$.
{}
\item[A1] $\nu$ is strictly increasing.

\item[A2] For all $0<x<y$, $\nu(y)-\nu(x)<\nu(-x)-\nu(-y)$.{}

\item[A3] $\nu''(x)\geq 0$ for $x<0$ and $\nu''(x)\leq 0$ for $x>0$.

\item[A4] $\nu'(0+)/\nu'(0-)<1${}

\end{itemize}
One can easily check that $\nu$ as in Assumption \ref{mienk} satisfies A0--A3 above. 

In \cite{paolo-andrea1} and \cite{paolo-andrea2} the specification $\nu_{0}(x)=\alpha_{1} x$, $x<0$, $\nu_{0}(x)=\alpha_{2} x$, $x\geq 0$
with $0<\alpha_{2}<\alpha_{1}$ was assumed.  
We deviate from both \cite{paolo-andrea1} and \cite{paolo-andrea2}.
We still take $\nu$ linear on the negative axis, but we assume it is twice continuously differentiable and bounded from above on the whole
of $\mathbb{R}$. Necessarily, we \emph{do not} assume A4 above.
Incorporating functions like $\nu_{0}$ is left for future research.}

\end{remark}

Let $B$ be an arbitrary $\sigma(\hat{\varepsilon}_{1},\ldots,\hat{\varepsilon}_{T})$-measurable random variable that will represent the
reference point of the investor.
Following \cite{kr2006}, we define the investor's overall satisfaction from $x$ dollars by
$$
U(x,B)=U(x)+\nu\bigl(U(x)-U(B)\bigr),\ x\in\mathbb{R}.
$$
This is easy to interpret: in addition to the direct utility $U(x)$ of $x$, the investor also evaluates, using the gain-loss function $\nu$, 
whether $U(x)$ exceeds or falls short of the reference utility $U(B)$.


\subsection{Personal equilibrium}

We now define the value function of the optimization problem we are dealing with. For each $\phi\in\Phi$, 
recalling $B({\phi})$ from \eqref{bebe}, 
 set
\begin{eqnarray}
u(x_{0},\phi):=\sup_{\psi\in\Phi}\mathbb{E}\bigl[U\bigl(W_{T}(x_{0},\psi),B(\phi)\bigr)\bigr].
\label{petitu}
\end{eqnarray}
A personal equilibrium is a portfolio $\phi^{\dagger}$ such that, choosing its ``independent copy'' $B(\phi^{\dagger})$ as reference point, 
the solution of the resulting optimization problem is just $\phi^{\dagger}$ itself. We formalize this heuristic description as follows.
 
\begin{definition} A strategy $\phi^{\dagger}\in\Phi$ is called a \emph{personal equilibrium} if
$$
\mathbb{E}\bigl[U\bigl(W_{T}(x_{0},\phi^{\dagger}),B(\phi^{\dagger})\bigr)\bigr]
=\sup_{\psi\in\Phi}\mathbb{E}\bigl[U\bigl(W_{T}(x_{0},\psi),B(\phi^{\dagger})\bigr)\bigr]
=u(x_{0},\phi^{\dagger}).
$$	
The set of personal equilibria is denoted by $\Phi^{\dagger}$.
Furthermore, $\phi^{\ddagger}\in\Phi^{\dagger}$ is called a \emph{preferred personal equilibrium} if 
$$
\mathbb{E}\bigl[U\bigl(W_{T}(x_{0},\phi^{\ddagger}),B(\phi^{\ddagger})\bigr)\bigr]=
\sup_{\phi^{\dagger}\in\Phi^{\dagger}}\mathbb{E}\bigl[U\bigl(W_{T}(x_{0},\phi^{\dagger}),B(\phi^{\dagger})\bigr)\bigr].
$$
The set of preferred personal equilibria is denoted by $\Phi^{\ddagger}$.	
\end{definition}

Preferred personal equilibria are thus the best-performing personal equilibria. It is not at all clear,
whether $\Phi^{\dagger}$ or $\Phi^{\ddagger}$ are non-empty. Our principal result answers these questions in the affirmative.

\begin{theorem}\label{main}
Let Assumptions \ref{filtra}, \ref{price1}, \ref{una}, \ref{utility} and \ref{mienk} be in force. Then
$\Phi^{\dagger}\neq \emptyset$ and, actually, a preferred personal equilibrium exists, i.e. $\Phi^{\ddagger}\neq \emptyset$.
\end{theorem}

\begin{remark}
{\rm We know from Theorem 3.2 of \cite{paolo-andrea1} that there is no uniqueness for personal equilibria. We do not know
if there is uniqueness for \emph{preferred} personal equilibria.}	
\end{remark}

\section{Proofs}\label{sec3}

\subsection{One-step case}

\label{secune}

In this subsection, we solve a one-step optimization problem that will later be recursively applied to construct an optimal strategy for an investor in a multi-step market with utility function $U(x,B)$, see Proposition \ref{dyna} below.

\begin{assumption}\label{onestep}
Let $1\leq t\leq T$ be given. Let a random variable $B$ be also given (this parametrizes the problem with respect to the
reference point).
We consider a $\mathcal{B}(\mathbb{R}^{t\times m})\otimes \mathcal{B}(\mathbb{R})$-measurable function
$$
\VB:\mathbb{R}^{t\times m}\times \mathbb{R}\to\mathbb{R}
$$
such that, for all $e\in\mathbb{R}^{t\times m}$, $x\mapsto \VB(e,x)$ is twice continuously differentiable and bounded from above by $C_{U}+C_{\nu}$. 

There are continuous functions $i_{V},j_{V},J_{V},\ell_{V},L_{V}$ such that  
$i_{V}:\mathbb{R}\to\mathbb{R}$ and $j_{V},J_{V},\ell_{V},L_{V}:\mathbb{R}\to (0,\infty)$, and  for all $e,x$,
\begin{equation}
\begin{aligned}
\label{ijlV}
i_{V}(x) &\leq \VB(e,x),\quad j_{V}(x)\leq \VB'(e,x)\leq J_{V}(x)\\
\ell_{V}(x) &\leq -\VB''(e,x)\leq L_{V}(x).
\end{aligned}
\end{equation}
Here $\VB',\VB''$ refer to derivatives with respect to $x$.

Furthermore, there are $\theta \in (0,\chi]$, and a continuous function $C_{V}:\mathbb{R}\to \mathbb{R}_{+}$ such that, 
for each $e,\bar{e}\in\mathbb{R}^{t\times m}$,
\begin{equation}\label{conii}
|\VB(e,x)-\VB(\bar{e},x)|\leq C_{V}(x)|e-\bar{e}|^{\theta}.
\end{equation}

Let $\varepsilon$ be an $\mathbb{R}^{m}$-valued random variable with $|\varepsilon|\leq C_{\varepsilon}$. 
There is a Borel function $f:\mathbb{R}^{t\times m}\to\mathbb{R}$ such that,
for all $o\in\mathbb{R}^{(t-1)\times m}$,
\begin{equation}\label{stronna}
\mathbb{P}[f(o,\varepsilon)\geq \alpha]\geq \alpha,\quad 
\mathbb{P}[f(o,\varepsilon)\leq -\alpha]\geq \alpha,
\end{equation}
and for all $e,\bar{e}\in\mathbb{R}^{t\times m}$,
\begin{equation}\label{efa}
|f(e)-f(\bar{e})|\leq C_{f}|e-\bar{e}|^{\chi},\quad \sup_{e\in \mathbb{R}^{t\times m}}|f(e)|\leq C_{f}.
\end{equation}
\end{assumption}
In \eqref{stronna}, we write $f(o,\varepsilon)$ instead of $f((o,\varepsilon))$ for the sake of simplicity, and we will use the same kind of notations in the rest of the paper. Remark that under Assumption \ref{onestep}, $x\mapsto \VB(e,x)$ is strictly increasing and strictly concave. Remark also that the functions $j_{V},J_{V},\ell_{V},L_{V}$ are continuous and strictly positive, hence
locally bounded away from $0$.

The above assumption will be in force throughout this subsection.
We introduce the two following functions, for all $(o,x,h)\in\mathbb{R}^{(t-1)\times m} \times\mathbb{R} \times \mathbb{R}$: 
\begin{eqnarray*}
\Gamma(B)(o,x,h) & := & \mathbb{E}\bigl[\VB \bigl(o,\varepsilon,x+hf(o,\varepsilon)\bigr)\bigr],\\
\gamma(B)(o,x,h) &:= & \mathbb{E}\bigl[\VB'\bigl(o,\varepsilon,x+hf(o,\varepsilon)\bigr)f(o,\varepsilon)\bigr].
\end{eqnarray*}

\begin{lemma}\label{well} 
Under Assumption \ref{onestep}, 
$\gamma(B)$ and $\Gamma(B)$ are well-defined Carath\'eodory integrands, i.e.\ 
for all $(x,h)$, $o \mapsto \gamma(B)(o,x,h),\Gamma(B)(o,x,h)$ are $\mathcal{B}(\mathbb{R}^{(t-1)\times m})$-measurable functions, 
and $(x,h)\mapsto \gamma(B)(o,x,h),\Gamma(B)(o,x,h)$ are continuous functions for all $o$. 
Moreover, for all $o$, the function $(x,h)\mapsto \Gamma(B)(o,x,h)$ is twice continuously differentiable and 
\begin{eqnarray}
\label{derivex}
\partial_{x}\Gamma(B)(o,x,h) & = & \mathbb{E}\bigl[\VB'\bigl(o,\varepsilon,x+hf(o,\varepsilon)\bigr)\bigr]\\
\label{deriveh}
\partial_{h}\Gamma(B)(o,x,h) & = & \mathbb{E}\bigl[\VB'\bigl(o,\varepsilon,x+hf(o,\varepsilon)\bigr)f(o,\varepsilon)\bigr]=\gamma(B)(o,x,h)\\
\label{derivexh}
\partial^2_{xh}\Gamma(B)(o,x,h) & = & \partial_{x}\gamma(B)(o,x,h)=\mathbb{E}\bigl[\VB''\bigl(o,\varepsilon,x+hf(o,\varepsilon)\bigr)f(o,\varepsilon)\bigr]\\
\label{derivehh}
\partial^2_{hh}\Gamma(B)(o,x,h) & = & \partial_{h}\gamma(B)(o,x,h)=\mathbb{E}\bigl[\VB''\bigl(o,\varepsilon,x+hf(o,\varepsilon)\bigr)f^2(o,\varepsilon)\bigr].
\end{eqnarray}
\end{lemma}
\begin{proof} Since $i_{V}(x)\leq \VB(\cdot,x)\leq C_{U}+C_{\nu}$, $j_{V}(x)\leq \VB'(\cdot,x)\leq J_{V}(x)$,  $-L_{V}(x) \leq \VB''(e,x)\leq -\ell_{V}(x)$ and $f$ is bounded, the expectations above are well-defined (and finite).
Dominated convergence implies that for all $o$, the functions $(x,h)\mapsto \gamma(B)(o,x,h),\Gamma(B)(o,x,h)$ are continuous. Fix $x,h$. As $(o,e) \mapsto x+ h f(o,e)$ and $\VB$ are Borel measurable, 
$(o,e) \mapsto \VB(o,e,x+ h f(o,e))$ is Borel measurable and Fubini theorem as in \cite[Proposition 7.29]{BS} implies that 
$o \mapsto \Gamma(B)(o,x,h)=\mathbb{E}[\VB(o,\varepsilon,x+ h f(o,\varepsilon))]$ is Borel measurable. The same reasoning applies for 
 measurability in $o$ of $\gamma(B)$. 
 
 Fix $M,N>0$. 
 For all $(x,h)\in [-M,M]\times [-N,N]$, $e\mapsto \VB \bigl(o,e,x+hf(o,e)\bigr)$ is integrable with respect to 
 the law of $\varepsilon$ under 
 $\mathbb{P}$, since it is measurable and bounded by $\max( \sup_{z \in D}| i_{V}(z)|,C_{U}+C_{\nu})$, where $D=[-M-NC_{f}, M+NC_{f}]$. 
 Moreover, for all $e\in \mathbb{R}^{t\times m}$, $(x,h)\mapsto \VB\bigl(o,e,x+hf(o,e)\bigr)$ is differentiable and 
 \begin{eqnarray*}
|\partial_{x}\VB\bigl(o,e,x+hf(o,e)\bigr)| & = & |\VB'\bigl(o,e,x+hf(o,e)\bigr)| \leq \sup_{z \in D}J_{V}(z)\\
|\partial_{h}\VB\bigl(o,e,x+hf(o,e)\bigr)|& = & |\VB'\bigl(o,e,x+hf(o,e)\bigr)f(o,e)| \leq   C_{f} \sup_{z \in D}J_{V}(z)
\end{eqnarray*}
and these are constant bounds. 
Thus, dominated convergence implies that $(x,h)\to \Gamma(B)(o,x,h)$ is 
differentiable on $(-M,M)\times (-N,N)$, hence on $\mathbb{R}^{2}$ (as $M$ and $N$ are arbitrary), and \eqref{derivex} and \eqref{deriveh} hold true. 
The rest of the proof is similar, 
using the other bounds of Assumption \ref{onestep}. 
\end{proof}

We define, for all $B$, the function $v(B):\mathbb{R}^{(t-1)\times m}\times \mathbb{R}\to\mathbb{R}$ by
$$
\vb(o,x):=\sup_{h\in\mathbb{R}}\mathbb{E}\bigr[\VB\bigl(o,\varepsilon,x+hf(o,\varepsilon)\bigr)\bigr]={}
\sup_{h\in\mathbb{R}}\Gamma(B)(o,x,h),
$$
for $o\in\mathbb{R}^{(t-1)\times m}$, $x\in\mathbb{R}$. In the case $t=1,$ we define
$$
\vb(x):=\sup_{h\in\mathbb{R}}\mathbb{E}\bigr[\VB\bigl(\varepsilon,x+h f(\varepsilon)\bigr)\bigr].
$$

The following result forms the core of our arguments. It shows that if $V$ satisfies Assumption \ref{onestep}, then $v$ also satisfies it with a $\theta/2$ instead of $\theta$ in \eqref{conii}. 
\begin{proposition}\label{recur} Assume that Assumption \ref{onestep} holds true. There exist functions 
$C_{h}:\mathbb{R}\to (0,\infty)$ and 
$\hb:\mathbb{R}^{(t-1)\times m}\times\mathbb{R}\to\mathbb{R}$
such that $C_{h}$ is continuous,  does not depend on $B$, $\hb$ is $\mathcal{B}(\mathbb{R}^{(t-1)\times m})\otimes \mathcal{B}(\mathbb{R})$-measurable, and, for all $o,x$, 
\begin{equation}\label{hbound}
|\hb(o,x)|\leq C_{h}(x),	
\end{equation} 
and $\hb(o,x)$ is the unique number that satisfies
$$
\vb(o,x)=\mathbb{E}\bigl[\VB\bigl(o,\varepsilon,x+\hb(o,x)f(o,\varepsilon)\bigr)\bigr].
$$
Furthermore, for all $o,\bar{o}\in\mathbb{R}^{(t-1)\times m}$, 
\begin{equation}\label{ccsillag}
|\hb(o,x)-\hb(\bar{o},x)|\leq C_{h}(x)|o-\bar{o}|^{\theta/2}.
\end{equation}

The function $\vb$ is $\mathcal{B}(\mathbb{R}^{(t-1)\times m})\otimes \mathcal{B}(\mathbb{R})$-measurable, bounded from above by $C_{U}+C_{\nu}$, 
twice continuosly differentiable in its second variable; there is a continuous function $C_{v}:\mathbb{R}\to (0,\infty)$, that does not depend on $B$, 
such that for all $o,\bar{o}\in\mathbb{R}^{(t-1)\times m}$, 
\begin{eqnarray}\label{petitvlip}
|\vb(o,x)-\vb(\bar{o},x)|\leq C_{v}(x)|o-\bar{o}|^{\theta/2}.
\end{eqnarray}
There exist continuous functions $i_{v},j_{v},J_{v},\ell_{v},L_{v}$, that do not depend on $B$, such that  
$i_{v}:\mathbb{R}\to\mathbb{R}$ and $j_{v},J_{v},\ell_{v},L_{v}:\mathbb{R}\to (0,\infty)$, and  for all $o,x$,
\begin{equation}
\begin{aligned}
\label{ijlv}
i_{v}(x)&\leq\vb(o,x),\quad j_{v}(x)\leq \vb'(o,x)\leq J_{v}(x),\\
\ell_{v}(x)&\leq -\vb''(o,x)\leq L_{v}(x).
\end{aligned}
\end{equation}
\end{proposition}

\begin{proof}
We will divide this rather complex proof into several steps. First, as $\VB$ is bounded above by $C_{U}+C_{\nu}$, the same 
holds true for $\vb$. 

\smallskip{}

\noindent\textbf{Boundedness of optimizer sequences}

\smallskip{}

Fix $(o,x)\in \mathbb{R}^{(t-1)\times m}\times\mathbb{R}$. Let  $h \in \mathbb{R}$. 
Define 
$$B_{h}:=\{\omega\in\Omega: f\big(o,\varepsilon(\omega)\big)\mathrm{sgn}(h)\leq -\alpha\},$$ where 
$\mathrm{sgn}(h)=1$ if $h\geq 0$ and $\mathrm{sgn}(h)=-1$ else. 
 For $h$ such that
$|x|-\alpha |h|<0$, using that $\VB$ is nondecreasing and concave in $x$, we estimate,
\begin{eqnarray*}
 & &  \hspace*{-2cm} \mathbb{E}\bigl[\VB\bigl(o,\varepsilon,x +  hf(o,\varepsilon)\bigr)\bigr] \\
   & \leq &  C_{U}+C_{\nu}+\mathbb{E}\bigl[1_{B_{h}}\VB\bigl(o,\varepsilon,x+hf(o,\varepsilon)\bigr)\bigr]\\
 & \leq & C_{U}+C_{\nu}+\mathbb{E}[1_{B_{h}}\VB(o,\varepsilon,|x|-\alpha |h|)]\\ 
& \leq & C_{U}+C_{\nu}+\mathbb{E}\bigl[1_{B_{h}}[\VB(o,\varepsilon,0)+\VB'(o,\varepsilon,0)(|x|-\alpha |h|)]\bigr]\\
&  \leq &2(C_{U}+C_{\nu})+\mathbb{E}\bigl[1_{B_{h}}[j_{V}(0)(|x|-\alpha |h|)]\bigr] \\
&  =  &  2(C_{U}+C_{\nu})+\mathbb{P}[B_{h}]j_{V}(0)(|x|-\alpha |h|)\\
&  \leq & 2(C_{U}+C_{\nu})+\alpha j_{V}(0)(|x|-\alpha |h|),
\end{eqnarray*}
where we have used \eqref{stronna} for the last inequality. 

If $|h|\geq [2(C_{U}+C_{\nu})+|i_{V}(x)|+j_{V}(0)\alpha |x|]/(j_{V}(0)\alpha^{2})$ then
$$
2(C_{U}+C_{\nu})+\alpha j_{V}(0)(|x|-\alpha |h|)\leq -|i_{V}(x)|
$$
holds. 
Let us now set 
$$K(x):=\frac{|x|}{\alpha}+ \frac{2(C_{U}+C_{\nu})+|i_{V}(x)|+j_{V}(0)\alpha |x|}{j_{V}(0)\alpha^{2}}> 0.$$ 
It is clear from the assumptions that $K$ is continuous. 
If $|h| \geq K(x)$, then using \eqref{ijlV}, 
\begin{equation}\label{0}
\mathbb{E}\bigl[\VB\bigl(o,\varepsilon,x+hf(o,\varepsilon)\bigr)\bigr]\leq -|i_{V}(x)| \leq \mathbb{E}[\VB(o,\varepsilon,x)] \leq \vb(o,x).
\end{equation}

\noindent\textbf{Existence of optimizer for $\Gamma(o,x,\cdot)$ }

\smallskip{}

Fix $(o,x)\in \mathbb{R}^{(t-1)\times m}\times\mathbb{R}$.  Let $h_{n}(B)(o,x)\in\mathbb{R}$, $n\in\mathbb{N}$ be a sequence such that
\begin{equation}\label{hn} 
\mathbb{E}\bigl[\VB\bigl(o,\varepsilon,x+h_{n}(B)(o,x)f(o,\varepsilon)\bigr)\bigr]\to \vb(o,x),\ n\to\infty.{}
\end{equation}
By \eqref{0}, we may replace in \eqref{hn} $h_{n}(B)$ by $\tilde{h}_{n}(B):=h_{n}(B)1_{\{|h_n(B)|\leq K(x)\}}$. By compactness, there is a 
subsequence $\tilde{h}_{n(k)}(B)$, $k\in\mathbb{N}$
such that $\tilde{h}_{n(k)}(B)(o,x)\to \bar{h}(B)(o,x)$ for some $\bar{h}(B)(o,x)$. By Fatou's lemma and continuity of $ \VB$ in $x$,
\begin{eqnarray*}
\vb(o,x) & \leq & \mathbb{E}\bigl[\limsup_{k\to\infty}\VB\bigl(o,\varepsilon,x+\tilde{h}_{n(k)}(B)(o,x)f(o,\varepsilon)\bigr)\bigr]\\
 & = & 
\mathbb{E}\bigl[\VB\bigl(o,\varepsilon,x+\bar{h}(B)(o,x)f(o,\varepsilon)\bigr)\bigr].
\end{eqnarray*}
Since $\VB$ is strictly concave in $x$, such $\bar{h}(B)(o,x)$ is unique. 




\smallskip

\noindent\textbf{Differentiability of $\bar{h}$ and $\vb$}

\smallskip{}
For all $(o,x)\in \mathbb{R}^{(t-1)\times m}\times\mathbb{R},$ $\Gamma(B)\bigl(o,x,\cdot\bigr)$ is differentiable and
\eqref{deriveh} holds.
Since $\vb(o,x)=\sup_{h\in\mathbb{R}} \Gamma(B)(o,x,h)$ and this supremum is attained at $\bar{h}(B)(o,x)$, 
the derivative must be $0$ at this point: 
\begin{eqnarray}
\label{premierordre}
\gamma(B)\bigl(o,x,\bar{h}(B)(o,x)\bigr)=0.
\end{eqnarray}
Fix $o\in \mathbb{R}^{(t-1)\times m}.$ Now, we want to apply the implicit function theorem (see p. 150 of Zeidler \cite{zeidler}) in order to show that $\bar{h}(B)(o,\cdot)$ is differentiable. First, Lemma \ref{well}  shows that $\gamma(B)\bigl(o,\cdot, \cdot)$ is differentiable. So, to apply the implicit function theorem in $(x,\bar{h}(B)(o,x))$ for all $x\in \mathbb{R}$, we need to prove that 
$$\bigl|\partial_{h}
\gamma(B)\bigl(o,x,\bar{h}(B)(o,x)\bigr)\bigr|>0.$$
In fact, we show  that for all $(x,h)\in \mathbb{R} \times \mathbb{R}$, 
\begin{equation}\label{bi}
|\partial_{h}\gamma(B)(o,x,h)|\geq \alpha^{3}\inf_{y\in D(x)}\ell_{V}(y),
\end{equation}
where $\alpha>0$ is given in \eqref{stronna}, and \begin{equation}\label{intD}
D(x):=[x-K(x)C_{f},x+K(x)C_{f}]\end{equation}
Recalling \eqref{derivehh}, and using  $\ell_{V}>0$ and \eqref{stronna}, we obtain that
\begin{eqnarray*} 
|\partial_{h}\gamma(B)(o,x,h)| & = & -\mathbb{E}\bigl[\VB''\bigl(o,\varepsilon,x+hf(o,\varepsilon)\bigr)f^2(o,\varepsilon)\bigr] \\
& \geq  & \mathbb{E}\bigl[\ell_{V}\bigl(x+hf(o,\varepsilon)\bigr)f^2(o,\varepsilon)\bigr] \\
& \geq  &\inf_{y\in D(x)}\ell_{V}(y)\mathbb{E}\bigl[f^2(o,\varepsilon)\bigr] \geq \inf_{y\in D(x)}\ell_{V}(y)
\mathbb{E}\bigl[f^2(o,\varepsilon)1_{\{f(o,\varepsilon)\geq \alpha\}}\bigr] \\
&\geq &  \alpha^{2}\inf_{y\in D(x)}\ell_{V}(y) 
\mathbb{P}[f(o,\varepsilon)\geq \alpha]\\
&\geq &  \alpha^{3}\inf_{y\in D(x)}\ell_{V}(y).
\end{eqnarray*}
Since $\ell_{V}$ is strictly positive, the conditions of the implicit function theorem
are met in every point $x$, and there exist $\delta(o)(x)>0$ (recall that we have fixed $o$), a continuously differentiable function
$\hat{h}(B)(o):(x-\delta(o)(x),x+\delta(o)(x))\to\mathbb{R}$ such that 
for all $y \in (x-\delta(o)(x),x+\delta(o)(x))$, 
$$
\gamma(B)\bigl(o,y,\hat{h}(B)(o)(y)\bigr)=0.
$$
Now, by the unicity of the root of \eqref{premierordre}, we necessarily have that for all $y \in (x-\delta(o)(x),x+\delta(o)(x))$
$$
\bar{h}(B)(o,y)=\hat{h}(B)(o)(y). 
$$
So $\bar{h}(B)(o,\cdot)$ is 
continuously differentiable in a neighbourhood of $x$ (which depends of $o$). Since this argument works for all $x$, 
$\bar{h}(B)(o,\cdot)$ is continuously differentiable on the whole real line with the derivative given by
\begin{eqnarray}
\label{moche} \partial_{x}\bar{h}(o,x) & = & -\frac{\partial_{x}
\gamma(B)\bigl(o,x,\bar{h}(B)(o,x)\bigr)}{\partial_{h}
\gamma(B)\bigl(o,x,\bar{h}(B)(o,x)\bigr)}.
\end{eqnarray}
As $o$ was arbitrary in $\mathbb{R}^{(t-1)\times m},$ \eqref{moche} holds for all $o \in \mathbb{R}^{(t-1)\times m}.$ \\
As $\vb(o,x)= \Gamma(B)(o,x,\bar{h}(o,x))$ and $x\mapsto \bar{h}(o,x), \Gamma(B)(o,x,h)$ are differentiable, $x\mapsto \vb(o,x)$ is also differentiable, and 
\begin{eqnarray}
\nonumber
\vb'(o,x) & = & \partial_{x}
\Gamma(B)\bigl(o,x,\bar{h}(B)(o,x)\bigr)
+  \partial_{h}
\Gamma(B)\bigl(o,x,\bar{h}(B)(o,x)\bigr) \partial_{x}\bar{h}(o,x) \\
& = & \nonumber
\partial_{x} \Gamma(B)\bigl(o,x,\bar{h}(B)(o,x)\bigr)\\
& =& 
\mathbb{E}\bigl[\VB'\bigl(o,\varepsilon,x+\bar{h}(o,x)f(o,\varepsilon)\bigr)\bigr],
\label{beau} 
\end{eqnarray}
recalling \eqref{deriveh}, \eqref{premierordre} and \eqref{derivex}.

\smallskip{}

\noindent\textbf{An estimate for $\Gamma$}

\smallskip{}

We claim that there is a continuous
function $A>0$ such that, for all $x \in\mathbb{R}$,  $|h|\leq K(x)$, and ${o},\bar{o}\in \mathbb{R}^{(t-1)\times m}$ 
\begin{eqnarray}
\label{ouf}
|\Gamma(B)(o,x,h)-\Gamma(B)(\bar{o},x,h)|\leq A(x)|o-\bar{o}|^{\theta}.
\end{eqnarray}
Recalling the interval $D(x)$ from \eqref{intD}, we can estimate
\begin{equation}\label{upper_control}
|\Gamma(B)(o,x,h)|\leq C_{U}+C_{\nu}+\sup_{y\in D(x)}|i_{V}(y)|,
\end{equation}
by \eqref{ijlV} and by $V(B)\leq C_{U}+C_{\nu}$.
Furthermore,
\begin{eqnarray*}
& & 
\hspace*{-2cm} |\Gamma(B)(o,x,h)-\Gamma(B)(\bar{o},x,h)|\\ 
&\leq&  \mathbb{E}\bigl[\bigl|\VB(o,\varepsilon,x+hf(o,\varepsilon))-\VB\bigl(\bar{o},\varepsilon,x+hf(o,\varepsilon)\bigr)\bigr|\bigr]\\
& &+	\mathbb{E}\bigl[\bigl|\VB\bigl(\bar{o},\varepsilon,x+hf(o,\varepsilon)\bigr)-\VB\bigl(\bar{o},\varepsilon,x+hf(\bar{o},\varepsilon)\bigr)\bigr|\bigr]\\
&\leq&\mathbb{E}\bigl[ C_{V}\bigl(x+hf(o,\varepsilon)\bigr)\bigr]|o-\bar{o}|^{\theta} \\
& & + \mathbb{E}\bigl[ \sup_{y\in D(x)}\VB'(\bar{o},\varepsilon,y)|h|\bigr|f(o,\varepsilon)-f(\bar{o},\varepsilon)\bigr|\bigr]\\
&\leq& \sup_{y\in D(x)}C_{V}(y)|o-\bar{o}|^{\theta}+\sup_{y\in D(x) }J_{V}(y)K(x)C_{f}|o-\bar{o}|^{\chi} \\
&\leq& \bigl(\sup_{y\in D(x)}C_{V}(y)+\sup_{y\in D(x) }J_{V}(y)K(x)C_{f}\bigr)\bigl(|o-\bar{o}|^{\theta}+|o-\bar{o}|^{\chi}\bigr)\\
&\leq & A(x) |o-\bar{o}|^{\theta},
\end{eqnarray*}	
where $$A(x):= 2\bigl(\sup_{y\in D(x)}C_{V}(y)+\sup_{y\in D(x) }J_{V}(y)K(x)C_{f}\bigr)+2\bigl(C_{U}+C_{\nu}+
\sup_{y\in D(x)}|i_{V}(y)|\bigr) >0,$$
recalling \eqref{upper_control} and 
using Lemma \ref{upi} with the choice $n=2$, 
$\theta_{1}=\theta$, $\theta_{2}=\chi$, since $0<\theta\leq \chi \leq 1$. 
The function $A$ is continuous as $C_{V},J_{V},K,i_{V}$ are, see Lemma \ref{cont}.


\noindent\textbf{Estimates of $v$ and its derivatives}

\smallskip{}
Let $(o,x)\in \mathbb{R}^{(t-1)\times m} \times  \mathbb{R}.$  
Clearly, $\vb(o,x)\geq \mathbb{E}[\VB(o,\varepsilon,x)]\geq i_{V}(x)$ so we may set $i_{v}(x):=i_{V}(x)$.
Furthermore, recalling \eqref{beau} and noting that $\VB'$ is non-increasing in $x$,
\begin{eqnarray*}
\vb'(o,x) &=& \mathbb{E}\bigl[\VB'\bigl(o,\varepsilon,x+\bar{h}(o,x)f(o,\varepsilon)\bigr)\bigr]\\ &\leq& \VB'(o,\varepsilon,x-K(x)C_{f})\leq J_{V}(x-K(x)C_{f})=:J_{v}(x)	
\end{eqnarray*}
and
$$
\vb'(o,x)\geq \VB'(o,\varepsilon,x+K(x)C_{f})\geq j_{V}(x+K(x)C_{f})=:j_{v}(x). 
$$
With these definitions,  the functions $i_{v},j_{v},J_{v}$  are continuous (recall that $K$ is continuous), 
$i_{v}:\mathbb{R}\to\mathbb{R}$ and $j_{v},J_{v}:\mathbb{R}\to (0,\infty)$.

Recalling  \eqref{beau}, we prove as in Lemma \ref{well} that 
\begin{equation}
\label{pasbeau}
\vb''(o, x) =  
\mathbb{E}\bigl[\VB''\bigl(o,\varepsilon,x+\bar{h}(o,x)f(o,\varepsilon)\bigr)\bigl(1+ f(o,\varepsilon)\partial_{x}\bar{h}(o,x) \bigr)\bigr],
\end{equation}
and that $x\mapsto \vb''(o, x)$ is continuous (here, \eqref{moche} shows that $x \mapsto \partial_{x}\bar{h}(o,x)$ is continuous). 
Recall again the interval $D(x)$ from \eqref{intD}. Let the probability 
$\mathbb{Q}$ be defined as follows (see \eqref{ijlV}), 
\begin{eqnarray*}
q & := & -\mathbb{E}\bigl[\VB''\bigl(o,\varepsilon,x+\bar{h}(o,x)f(o,\varepsilon)\bigr)\bigr] \geq \inf_{y \in D(x)} \ell (y)>0 \\
\frac{d\mathbb{Q}}{d\mathbb{P}} &:= & \frac{-\VB''\bigl(o,\varepsilon,x+\bar{h}(o,x)f(o,\varepsilon)\bigr)}{q}.
\end{eqnarray*}
We denote by $\mathbb{E}_{\mathbb{Q}}$ the expectation under $\mathbb{Q}$.  
Recalling \eqref{moche},  \eqref{derivexh} and \eqref{derivehh}, we estimate that
\begin{eqnarray}
\nonumber
-\vb''(o, x) & = &   q\mathbb{E}_{\mathbb{Q}}\bigl[1+ f(o,\varepsilon)\partial_{x}\bar{h}(o,x) \bigr] \\
\nonumber
& = &  q\mathbb{E}_{\mathbb{Q}}\Bigl[1- f(o,\varepsilon)\frac{\mathbb{E}\bigl[\VB''\bigl(o,\varepsilon,x+\bar{h}(o,x)f(o,\varepsilon)\bigr)f(o,\varepsilon)\bigr]}{\mathbb{E}\bigl[\VB''\bigl(o,\varepsilon,x+\bar{h}(o,x)f(o,\varepsilon)\bigr)f^2(o,\varepsilon)\bigr]}\Bigr]\\
\nonumber
& = &  q\mathbb{E}_{\mathbb{Q}}\Bigl[1- f(o,\varepsilon)\frac{\mathbb{E}_{\mathbb{Q}}[f(o,\varepsilon)]}{\mathbb{E}_{\mathbb{Q}}[f^2(o,\varepsilon)]}\Bigr]={}
q\frac{\mathbb{E}_{\mathbb{Q}}[f^2(o,\varepsilon)] -\mathbb{E}^2_{\mathbb{Q}}[f(o,\varepsilon)]}{\mathbb{E}_{\mathbb{Q}}[f^2(o,\varepsilon)]}\\
\nonumber
& = & q\frac{\mathbb{E}_{\mathbb{Q}}\bigl[ \bigl(f(o,\varepsilon)-\mathbb{E}_{\mathbb{Q}}[f(o,\varepsilon)]\bigr)^2\bigr]}{\mathbb{E}_{\mathbb{Q}}[f^2(o,\varepsilon)]}.
\end{eqnarray}
We now distinguish between two cases. If $\mathbb{E}_{\mathbb{Q}}[f(o,\varepsilon)]>0$, then 
\begin{eqnarray*}
-\vb''(o, x)& \geq  & q\frac{\mathbb{E}_{\mathbb{Q}}\bigl[ \bigl(f(o,\varepsilon)-\mathbb{E}_{\mathbb{Q}}
[f(o,\varepsilon)]\bigr)^2 1_{\{f(o,\varepsilon) \leq -\alpha\}}\bigr]}{\mathbb{E}_{\mathbb{Q}}[f^2(o,\varepsilon)]} \\
& \geq &  q\frac{\mathbb{E}_{\mathbb{Q}}\bigl[ 1_{\{f(o,\varepsilon) \leq -\alpha\}} \alpha^2\bigr]}{\mathbb{E}_{\mathbb{Q}}[f^2(o,\varepsilon)]} \geq 
\frac{q\alpha^2}{C_f^2} \mathbb{E}_{\mathbb{Q}}\bigl[ 1_{\{f(o,\varepsilon) \leq -\alpha\}} \bigr]\\
& \geq &  \frac{\alpha^2}{C_f^2}\mathbb{E}\bigl[ -\VB''\bigl(o,\varepsilon,x+\bar{h}(o,x)f(o,\varepsilon)\bigr)1_{\{f(o,\varepsilon) \leq -\alpha\}} \bigr] \\
& \geq &  \frac{\alpha^2}{C_f^2} \inf_{y\in D(x)} \ell_{V}(y)\mathbb{E}\bigl[ 1_{\{f(o,\varepsilon) \leq -\alpha\}} \bigr] \geq \frac{\alpha^3}{C_f^2} \inf_{y\in D(x)} \ell_{V}(y),
\end{eqnarray*}
using \eqref{stronna}. Now,  if $\mathbb{E}_{\mathbb{Q}}[f(o,\varepsilon)]\leq 0$, then 
\begin{eqnarray*}
-\vb''(o, x)& \geq  & q\frac{\mathbb{E}_{\mathbb{Q}}\bigl[ \bigl(f(o,\varepsilon)-\mathbb{E}_{\mathbb{Q}}[f(o,\varepsilon)]\bigr)^2 
1_{\{f(o,\varepsilon) \geq \alpha\}}\bigr]}{\mathbb{E}_{\mathbb{Q}}[f^2(o,\varepsilon)]} \\
&  \geq &   q
\frac{\mathbb{E}_{\mathbb{Q}}\bigl[ 1_{\{f(o,\varepsilon) \geq \alpha\}} \alpha^2\bigr]}{\mathbb{E}_{\mathbb{Q}}[f^2(o,\varepsilon)]} 
\\
& \geq &  \frac{\alpha^2}{C_f^2}\mathbb{E}\bigl[ -\VB''\bigl(o,\varepsilon,x+\bar{h}(o,x)f(o,\varepsilon)\bigr)1_{\{f(o,\varepsilon) \geq \alpha\}} \bigr] \\
& \geq  & 
\frac{\alpha^3}{C_f^2} \inf_{y\in D(x)} \ell_{V}(y):= \ell_{v}(x).
\end{eqnarray*}
Then, $\ell_{v}:\mathbb{R}\to (0,\infty)$ is continuous, see Lemma \ref{cont}. 
For the upper bound, using \eqref{moche}, \eqref{derivexh} and  \eqref{bi}, we get that 
\begin{eqnarray*}
\left| \partial_{x}\bar{h}(o,x)\right| & = & \Bigl|\frac{\partial_{x}
\gamma(B)\bigl(o,x,\bar{h}(B)(o,x)\bigr)}{\partial_{h}
\gamma(B)\bigl(o,x,\bar{h}(B)(o,x)\bigr)}\Bigr|\\
& \leq &  \frac{\mathbb{E}\bigl[|\VB''\bigl(o,\varepsilon,x+\bar{h}(B)(o,x)f(o,\varepsilon)\bigr)f(o,\varepsilon)|\bigr]}{\alpha^{3}\inf_{y\in D(x)}\ell_{V}(y)}\\
& \leq & \frac{C_{f}\sup_{y\in D(x)}L_{V}(y)}{\alpha^{3}\inf_{y\in D(x)}\ell_{V}(y)}=\frac{\sup_{y\in D(x)}L_{V}(y)}{C_{f}\ell_{v}(x)}.
\end{eqnarray*}
Recalling \eqref{pasbeau}, we get that 
\begin{eqnarray*}
-\vb''(o, x) & = & \mathbb{E}\bigl[-\VB''\bigl(o,\varepsilon,x+\bar{h}(o,x)f(o,\varepsilon)\bigr)\bigl(1+ f(o,\varepsilon)\partial_{x}\bar{h}(o,x) \bigr)\bigr] \\
& \leq  & \sup_{y\in D(x)}L_{V}(y)\left(1+ \frac{\sup_{y\in D(x)}L_{V}(y)}{\ell_{v}(x)} \right) =:L_{v}(x),
\end{eqnarray*}
and $L_{v}:\mathbb{R}\to (0,\infty)$ is continuous using again Lemma \ref{cont}. 

\noindent\textbf{Continuity of $\bar{h}$ with respect to past}

\smallskip{}

Let $x\in \mathbb{R}.$    In this part of the proof, we suppress dependence on $B$ in the notation, for simplicity.
Let $|h|\leq K(x)$ hold from now on. Recall the interval $D(x)$ from \eqref{intD} again.

Let $o,\bar{o} \in \mathbb{R}^{(t-1)\times m}$. Let $h\in \mathbb{R}$ and $\bar{h}(o,x)$ (resp. $\bar{h}(\bar{o},x)$) be the optimizer of $\Gamma(B)(o,x,\cdot)$  (resp. 
$\Gamma(B)(\bar{o},x,\cdot)$). 
One can write, by the Newton-Leibniz rule:
\begin{equation}\label{nl}
\Gamma(o,x,h)-\Gamma\bigl(o,x,\bar{h}(o,x)\bigr)=\int_{\bar{h}(o,x)}^{h}\gamma(o,x,\xi)\, d\xi.
\end{equation}
The first order condition \eqref{premierordre} and \eqref{bi} imply that for any $\xi$, 
\begin{eqnarray*}
|\gamma(o,x,\xi)|=|\gamma(o,x,\xi)-\gamma\bigl(o,x,\bar{h}(o,x)\bigr)| & \geq &    |\xi-\bar{h}(o,x)|\alpha^{3}\inf_{y\in D(x)}\ell_{V}(y)\\
 &= & |\xi-\bar{h}(o,x)|C_f^{2}\ell_{v}(x). 
\end{eqnarray*}
First assume that $\bar{h}(o,x) \leq h.$ Let $\xi$ such that $\bar{h}(o,x)\leq \xi\leq h$. Then, as $\bar{h}(o,x)$ is the maximum of $h \mapsto \Gamma(o,x,h)$ and 
$\partial_{h}\Gamma(o,x,h) =\gamma(o,x,h)$, we get that 
$\gamma(o,x,h) \leq 0$ and 
\begin{equation}\label{gammabecs1}
-\gamma(o,x,\xi)\geq |\xi-\bar{h}(o,x)|C_f^{2}\ell_{v}(x)=\bigl(\xi-\bar{h}(o,x)\bigr)C_f^{2}\ell_{v}(x),	
\end{equation}
so \eqref{nl} implies
\begin{equation}\label{als}
\Gamma(o,x,h)-\Gamma\bigl(o,x,\bar{h}(o,x)\bigr)\leq -\frac{C_f^{2}\ell_{v}(x)}{2}(h-\bar{h}(o,x))^{2}.
\end{equation}
Assume now that $\bar{h}(o,x)>h$. Let $\xi$ such that $\bar{h}(o,x)\geq \xi> h$. Then $\gamma(o,x,h) \geq 0$ and
\begin{equation}\label{gammabecs2}
\gamma(o,x,\xi)\geq |\xi-\bar{h}(o,x)|C_f^{2}\ell_{v}(x)=-(\xi-\bar{h}(o,x))C_f^{2}\ell_{v}(x),	
\end{equation}
leading again to \eqref{als}. For $o,\bar{o}$ \eqref{als} gives for $h=\bar{h}(\bar{o},x)$,
\begin{eqnarray*}
 \frac{C_f^{2}\ell_{v}(x)}{2}\big(\bar{h}({o},x)-\bar{h}(\bar{o},x)\big)^{2} & \leq & |\Gamma\bigl(o,x,\bar{h}(o,x)\bigr)-\Gamma\bigl(o,x,\bar{h}(\bar{o},x)\bigr)| \\
& \leq & 
|\Gamma\bigl(o,x,\bar{h}(o,x)\bigr)-\Gamma\bigl(\bar{o},x,\bar{h}(\bar{o},x)\bigr)|  \\
& & + |\Gamma\bigl(o,x,\bar{h}(\bar{o},x)\bigr)-\Gamma\bigl(\bar{o},x,\bar{h}(\bar{o},x)\bigr)| \\
& \leq & 
|\Gamma\bigl(o,x,\bar{h}(o,x)\bigr)-\Gamma\bigl(\bar{o},x,\bar{h}(\bar{o},x)\bigr)|  + A(x) |o-\bar{o}|^{\theta},
\end{eqnarray*}
where for the last inequality, we have used \eqref{ouf} as $|\bar{h}(\bar{o},x)| \leq K(x).$ 
Recalling 
\eqref{0}:
$$
v(o,x)=\sup_{h\in\mathbb{R}}\Gamma(o,x,h)=\sup_{|h|\leq K(x)}\Gamma(o,x,h)=\Gamma(o,x,\bar{h}(o,x)),
$$
and the same holds true for $\bar{o}$. It follows that 
\begin{eqnarray*}
|\Gamma\bigl(o,x,\bar{h}(o,x)\bigr)-\Gamma\bigl(\bar{o},x,\bar{h}(\bar{o},x)\bigr)| 
& = &  
\bigl|\sup_{|h|\leq K(x)}\Gamma(o,x,h)-\sup_{|h|\leq K(x)}\Gamma(\bar{o},x,h)\bigr| \\
& \leq  &  
\sup_{|h|\leq K(x)}\bigl|\Gamma(o,x,h)-\Gamma(\bar{o},x,h)\bigr| \\
& \leq  &   A(x) |o-\bar{o}|^{\theta},
\end{eqnarray*}
where 
the last inequality follows from \eqref{ouf} as every $h$ is such that $|h| \leq K(x).$ So, we get that 
$$
|\bar{h}(o,x)-\bar{h}(\bar{o},x)|^{2}\leq \frac{4A(x)|o-\bar{o}|^{\theta}}{C_f^{2}\ell_{v}(x)},
$$
which implies that, indeed, \eqref{ccsillag} is valid with 
$$
C_{h}(x):=K(x)+ \frac{2}{C_f}\sqrt{\frac{A(x)}{\ell_{v}(x)}}>0.
$$
As $A,K, \ell_{v}$ are continuous, so is  $C_{h}$. Remark that $C_{h}$ does not depend on $B$.

\noindent\textbf{Continuity of $v$ with respect to the past}

\smallskip{}
Let $x\in \mathbb{R}.$ Let $o,\bar{o} \in \mathbb{R}^{(t-1)\times m}$. Let $\bar{h}(o,x)$ (resp. $\bar{h}(\bar{o},x)$) be the optimizer of $\Gamma(B)(o,x,\cdot)$  (resp. 
$\Gamma(B)(\bar{o},x,\cdot)$). 
We have already estimated that (see \eqref{ijlv})
\begin{equation}\label{v_upper}
|\vb(o,x)|\leq C_{U}+C_{\nu}+|i_{v}(x)|.
\end{equation}
Note that $|\bar{h}(o,x)f(o,\varepsilon)|\leq K(x)C_{f}$. Recall \eqref{ijlV} and \eqref{efa}. Estimate
\begin{eqnarray*}
& & \hspace*{-1cm} |\vb(o,x)-\vb(\bar{o},x)|\\
&\leq& \mathbb{E}\bigl[\bigl|\VB\bigl(o,\varepsilon,x+\hb(o,x)f(o,\varepsilon)\bigr)-\VB\bigl(\bar{o},\varepsilon,x+\hb(\bar{o},x)
f(\bar{o},\varepsilon)\bigr)\bigr|\bigr]\\
&\leq& 	
\mathbb{E}\bigl[\bigl|\VB\bigl(o,\varepsilon,x+\hb(o,x)f(o,\varepsilon)\bigr)-\VB\bigl(\bar{o},\varepsilon,x+\hb(o,x)
f(o,\varepsilon)\bigr)\bigr|\bigr]\\
& &  + 
\mathbb{E}\bigl[\bigl|\VB\bigl(\bar{o},\varepsilon,x+\hb(o,x)f(o,\varepsilon)\bigr)-\VB\bigl(\bar{o},\varepsilon,x+\hb(\bar{o},x)
f(\bar{o},\varepsilon)\bigr)\bigr|\bigr]\\
&\leq& \sup_{y\in D(x)}C_{V}(y)|o-\bar{o}|^{\theta} \\
& &  + 
\mathbb{E}\bigl[\sup_{y\in D(x)}V(B)'\bigl(\bar{o},\varepsilon,y)
|\hb(o,x)f(o,\varepsilon)- \hb(\bar{o},x)f(\bar{o},\varepsilon)| \bigr]\\
&\leq& \sup_{y\in D(x)}C_{V}(y)|o-\bar{o}|^{\theta}
\\ & & + \sup_{y\in D(x)}J_{V}(y)
\Bigl[K(x)\mathbb{E}\bigl[|f(o,\varepsilon)-f(\bar{o},\varepsilon)|\bigr]+C_{f}|\hb(o,x)-\hb(\bar{o},x)|\Bigr]\\
&\leq& \bigl[\sup_{y\in D(x)}C_{V}(y)+\sup_{y\in D(x)}J_{V}(y)C_{f}[K(x)+C_{h}(x)]\bigr]
\\ &\times&
[|o-\bar{o}|^{\theta}+|o-\bar{o}|^{\chi}+|o-\bar{o}|^{\theta/2}],
\end{eqnarray*}
where we have used 
\begin{eqnarray*}
|\hb(o,x)f(o,\varepsilon)-\hb(\bar{o},x)
f(\bar{o},\varepsilon)|
 & \leq  & |\hb(o,x)f(o,\varepsilon)-\hb(o,x)
f(\bar{o},\varepsilon)| 
\\ & & 
+|\hb(o,x)
f(\bar{o},\varepsilon)-\hb(\bar{o},x)
f(\bar{o},\varepsilon)|.\end{eqnarray*}
Recalling \eqref{v_upper} and Lemma \ref{upi} with the choice $n=3$, 
$\theta_{1}:=\theta/2$, $\theta_{2}:=\theta$, $\theta_{3}=\chi$,{}
we may set
$$
C_{v}(x):=3\bigl[\sup_{y\in D(x)}C_{V}(y)+\sup_{y\in D(x)}J_{V}(y)C_{f}[K(x)+C_{h}(x)]\bigr]+2\bigl(
C_{U}+C_{\nu}+|i_{v}(x)|\bigr),
$$
and \eqref{petitvlip} holds. 
Since $C_{V},J_{V},K,C_{h},i_{v}$ are continuous,
so is $C_{v}(x),$ see Lemma \ref{cont}. 

\smallskip{}

\noindent\textbf{Measurability}

\smallskip{}

It is known that Carath\'{e}odory integrand (i.e. a function of two variables that is measurable in the first and continuous in the second) is jointly 
measurable, see \cite[Lemma 4.51]{AB}. 
So, the function $\Gamma(B)$ is 
$\mathcal{B}(\mathbb{R}^{(t-1)\times m})\otimes \mathcal{B}(\mathbb{R})\otimes \mathcal{B}(\mathbb{R})$-measurable, 
see Lemma \ref{well} (the first variable is here $o$ and the second $(x,h)$). Now, $\bar{h}(B)$ is continuous in $o$ 
(see \eqref{ccsillag}) and we have proved that $\bar{h}(B)$ is differentiable in $x$ (see \eqref{moche}). 
So, $\bar{h}(B)$ is continuous in each 
variable separately, hence it is $\mathcal{B}(\mathbb{R}^{(t-1)\times m})\otimes\mathcal{B}(\mathbb{R})$-measurable.
Then so is $v(B)$, as $v(B)(o,x)=\Gamma(B)(o,x,\bar{h}(B)(o,x))$. 
Now our proof is complete.
\end{proof}

\subsection{Dynamic programming}

We prove that there exists some bounded and Hölder-continuous solution for the optimization problem \eqref{petitu}.

\begin{proposition}\label{dyna}
Let Assumptions \ref{filtra}, \ref{price1}, \ref{una}, \ref{utility} and \ref{mienk} hold. 
Let $x_{0} \in \mathbb{R}$ and $\phi\in\Phi$ be arbitrary. Then there exists a unique optimizer $\psi^{*}:=\psi^{*}(\phi)(\cdot,x_{0})\in\Phi$ such that
$$
u(x_{0},\phi)=\sup_{\psi\in\Phi}\mathbb{E}\bigl[U\bigl(W_{T}(x_{0},\psi),B(\phi)\bigr)\bigr]=\mathbb{E}\bigl[U\bigl(W_{T}(x_{0},\psi^{*}),B(\phi)\bigr)\bigr].
$$	
Denoting by $\bar{\psi}^{*}$ a Borel function associated to $\psi^{*}$ by Doob's theorem, i.e. 
$\psi^{*}_{t}:=\bar{\psi}^{*}_{t}(\varepsilon^{t-1})$, $1\leq t\leq T$ (we mean that $\bar{\psi}^{*}_1$ is constant), 
$\bar{\psi}^{*}$ can be chosen bounded and Hölder-continuous, where constants are independent of $\phi$. That is,
there exists  a continuous function $C: \mathbb{R} \to (0,\infty)$ such that for all $1\leq t \leq T$,  for all $e^{t-1},\bar{e}^{t-1}\in\mathbb{R}^{(t-1)\times m}$, 
\begin{eqnarray}\label{m1}
|\bar{\psi}^{*}(\phi)_t (e^{t-1},x_{0})| & \leq &  C(x_{0}) \\
\label{m2}
|\bar{\psi}^{*}(\phi)_{t}(e^{t-1},x_{0})-\bar{\psi}^{*}(\phi)_{t}(\bar{e}^{t-1},x_{0})| & \leq &  C(x_{0})|e^{t-1}-\bar{e}^{t-1}|^{\chi/2^{T-t+1}}.
\end{eqnarray}
Note again that the constant $C(x_{0})$ depends only on $x_{0}$ and neither on $B$ nor on $\phi$. 
\end{proposition}
\begin{proof}
We will apply the results of Section \ref{secune} recursively. Let $B$ be an arbitrary bounded random variable
that is independent of $\mathcal{F}_{T}^{\varepsilon}$. First, we define for all $x\in\mathbb{R}$
\begin{eqnarray*}
V_{T}(B)(e^{T},x) & := & \mathbb{E}\bigl[U(x,B)\bigr]=\mathbb{E}\Bigl[U(x)+\nu \bigl(U(x)-U(B)\bigr)\Bigr],\ e^{T}\in\mathbb{R}^{T\times m},\\
V_{t}(B)(e^{t},x) & := & \sup_{h\in\mathbb{R}}\mathbb{E}\bigl[V_{t+1}(B)\bigl(e^{t},\varepsilon_{t+1},x+h f_{t+1}(e^{t},\varepsilon_{t+1})\bigr)\bigr],\ e^{t}\in\mathbb{R}^{t\times m}.
\end{eqnarray*}
We check Assumption \ref{onestep} for $V(B)$ with $V(B)(e^{T},x):=U(x,B)$, and  then, for $V_{T}(B)$. 
We take $\varepsilon:=\varepsilon_{T}$ and $f(e^{T}):=f_{T}(e^{T})$, $e^{T}\in\mathbb{R}^{T\times m}$; \eqref{stronna} follows from
Assumption \ref{una}, and \eqref{efa} is true by Assumption \ref{price1}. 
Now, Assmptions \ref{utility}  and \ref{mienk} imply that $V(B)$ and $V_{T}(B)$ are bounded from above by $C_{U}+C_{\nu}$, and that $V(B)$ is twice continuously
differentiable in $x$. 

Note that neither $V(B)$ nor $V_{T}(B)$ depend on $e^{T}$. So,  \eqref{conii} is trivial with $C_{V}=0$ and 
$\theta=\chi$. As $U$ and $\nu$ are Borel,  $V(B)$  is 
 trivially 
$\mathcal{B}(\mathbb{R}^{T \times m})\otimes \mathcal{B}(\mathbb{R})$-measurable. 
Using Fubini theorem, $V_T(B)$  is also  
$\mathcal{B}(\mathbb{R}^{T \times m})\otimes \mathcal{B}(\mathbb{R})$-measurable. 

We now prove \eqref{ijlV} for $V(B)$. On the event $\{B\leq x\}$, $U(x,B)\geq U(x)+\nu(0)=U(x)$ while on the event $\{B>x\}$ we may estimate
$$U(x,B)= U(x)+k_{-}\bigl(U(x)-U(B)\bigr)\geq (1+k_{-})U(x) -k_{-}C_{U}.$$ 
Thus, we may set 
\begin{eqnarray}
i_{V}(x):=\min\{U(x),(1+k_{-})U(x) -k_{-}C_{U}\}=(1+k_{-})U(x) -k_{-}C_{U},
\label{borneinfU}
\end{eqnarray}
as $U\leq C_{U}$. 
We have that $U'(x,B)=U'(x)+\nu'\bigl(U(x)-U(B)\bigr)U'(x)$. So, $0\leq \nu' \leq k_{-}$ and $U' \geq 0$ imply 
\begin{eqnarray}
U'(x)\leq U'(x,B)\leq U'(x)+k_{-}U'(x)
\label{borneinfUprime}
\end{eqnarray}
and we may set $j_{V}(x):=U'(x)$ and $J_{V}(x):=(1+k_{-})U'(x)$.
We have that 
$$U''(x,B)=U''(x)+\nu''\bigl(U(x)-U(B)\bigr)\bigl(U'(x)\bigr)^2+\nu'\bigl(U(x)-U(B)\bigr)U''(x).$$ 
Furthermore, since $U''\leq 0,$ $-C_{\nu}\leq \nu''\leq0$  and $0\leq \nu'\leq k_{-}$,
\begin{eqnarray*}
U''(x)-C_{\nu}\bigl(U'(x)\bigr)^2  + k_{-}U''(x) 
\leq U''(x,B) \leq U''(x)
\end{eqnarray*}
so we may set $\ell_{V}(x):=-U''(x)$ and $L_{V}(x):=-(1+k_{-})U''(x)+C_{\nu}(U'(x))^2$. Assumption \ref{utility} implies that $i_{V},j_{V},J_{V},\ell_{V},L_{V}$ are continuous and  that  
$i_{V}:\mathbb{R}\to\mathbb{R}$ and $j_{V},J_{V},\ell_{V},L_{V}:\mathbb{R}\to (0,\infty)$. Note that, as these functions do not depend on $B$, the same bounds work for 
$V_{T}(B)$, which prove \eqref{ijlV} for $V_{T}(B)$. 

So,  Assumption \ref{onestep} holds for $V(B)$, and we can apply Lemma  \ref{well}. As 
$$\Gamma(B)(e^{T-1},x,0)=\mathbb{E}\bigl[\VB \bigl(e^{T-1},\varepsilon_T,x\bigr)\bigr]=\mathbb{E}\bigl[U(x,B)\bigr]= V_{T}(B)(e^{T},x),$$ 
$x\mapsto V_{T}(B)(e^{T},x)$ is twice continuously differentiable,  and $V_{T}(B)$ also satisfies Assumption \ref{onestep}.

Setting $B=B(\phi)$ now, $V_{T}(B(\phi))$ also satisfies Assumption \ref{onestep} for an arbitrary $\phi\in\Phi$. For simplicity, we 
don\rq{}t write the dependence of $B$ on $\phi$ until \eqref{optim}. Proposition \ref{recur} for  $V_{T}(B)$ implies 
that there exist some functions 
$C_{T}:\mathbb{R}\to (0,\infty)$ and 
$\hb_T:\mathbb{R}^{(T-1)\times m}\times\mathbb{R}\to\mathbb{R}$
such that $C_{T}$ is continuous,  $\hb_T$ is $\mathcal{B}(\mathbb{R}^{(T-1)\times m})\otimes \mathcal{B}(\mathbb{R})$-measurable, and, for all $e^{T-1},x$, 
$
|\hb_{T}(e^{T-1},x)|\leq C_{T}(x),	
$
and $\hb_{T}(e^{T-1},x)$ is the unique number that satisfies
$$
V_{T-1}(B)(e^{T-1},x)=\mathbb{E}\bigl[V_{T}(B)\bigl(e^{T-1},\varepsilon_T,x+\hb_{T}(e^{T-1},x)f_{T}(e^{T-1},\varepsilon_{T})\bigr)\bigr].
$$
Furthermore, for all $e^{T-1},\bar{e}^{T-1}\in\mathbb{R}^{(t-1)\times m}$, (recall that $\theta=\chi$ in \eqref{conii} for $V_{T}(B)$),
$$
|\hb_{T}(e^{T-1},x)-\hb_{T}(\bar{e}^{T-1},x)|\leq C_{T}(x)|e^{T-1}-\bar{e}^{T-1}|^{\chi/2}.
$$
Moreover, $V_{T-1}(B)$ satisfies Assumption \ref{onestep} with $\theta=\chi/2$ in \eqref{conii}. So, we can repeat the applications of Proposition \ref{recur},  
construct $C_{t}:\mathbb{R}\to (0,\infty)$ and 
$\hb_t:\mathbb{R}^{(t-1)\times m}\times\mathbb{R}\to\mathbb{R}$, and obtain the same properties for them (with $\theta=\chi/2^{T-t+1}$ in \eqref{conii}), and $V_{t}(B)$ for $1 \leq t \leq T$. 

Let $\bar{\psi}^{*}_{1}=\bar{\psi}^{*}_{1}(e_0,x_{0}):=\bar{h}(B)_{1}(x_{0})$ and define recursively
$$ 
\bar{\psi}^{*}_{t+1}(e^t,x_{0}):=\bar{h}(B)_{t+1}\Bigr(e^t,x_{0}+\sum_{j=1}^{t}\bar{\psi}^{*}_{j}(e^{j-1},x_{0})f_{j}(e^j)\Bigr),
$$
for $1\leq t\leq T-1$ and $e^t \in \mathbb{R}^{t\times m}$.

We prove by induction that  $|\bar{\psi}^{*}_{t}(e^{t-1},x_{0})|  \leq \bar{C}_{t}(x_{0})$ for all $e^{t-1} \in \mathbb{R}^{(t-1)\times m}$, for some continuous function $\bar{C}_{t}$. 
For $t=1$, just choose $\bar{C}_{1}={C}_{1}.$ Assume that the induction holds until $t$ with $1\leq t\leq T-2$. Then, 
\begin{eqnarray*}
|\bar{\psi}^{*}_{t+1}(e^t,x_{0})| & \leq & C_{t+1}\Bigr(x_{0}+\sum_{j=1}^{t}\bar{\psi}^{*}_{j}(e^{j-1},x_{0})f_{j}(e^j)\Bigr)\\
 & \leq & \sup_{y \in K_t(x_{0}) } C_{t+1}(y)=:\bar{C}_{t+1}(x_0)
\end{eqnarray*}
where $K_t(x)=[x-C_{f}\sum_{j=1}^{t} \bar{C}_{j}(x),x+ C_{f}\sum_{j=1}^{t} \bar{C}_{j}(x)]$. Lemma \ref{cont} shows that $\bar{C}_{t+1}$ is continous. 
Now, \eqref{m1} holds choosing $C(x_{0})=\max_{1\leq t\leq T}\bar{C}_{t}(x_0).$ It is clear that  $C$ is continuous. 
As the ${C}_{t}$ do not depend on $B$ (and thus on $\phi$), $C$ does not depend on $\phi$. 
For \eqref{m2}, just observe that 
\begin{eqnarray*}
 & &  \hspace*{-2cm} |
 \bar{\psi}^{*}_{t}(e^{t-1},x_{0})-\bar{\psi}^{*}_{t}(\bar{e}^{t-1},x_{0})|
  \\
 & \leq  & C_{t}\Bigl(x_{0}+\sum_{j=1}^{t-1}\bar{\psi}^{*}_{j}(e^{j-1},x_{0})f_{j}(e^j)\Bigr)|e^{t-1}-\bar{e}^{t-1}|^{\chi/2^{T-t+1}}\\
& \leq  & C(x_{0})|e^{t-1}-\bar{e}^{t-1}|^{\chi/2^{T-t+1}}.
\end{eqnarray*}

We finally establish that the strategy $\psi^{*}_{1}:=\bar{\psi}^{*}_{1}$ and 
$\psi^{*}_{t+1}:=\bar{\psi}^{*}_{t+1}(\varepsilon^t)$, $1\leq t\leq T-1$ 
is optimal, that is, $\psi^{*} \in \Phi$, and for all $\psi\in\Phi$,
\begin{equation}\label{optim}
\mathbb{E}\bigl[U\bigl(W_{T}(x_{0},\psi^{*}(\phi)),B(\phi)\bigr)\bigr]\geq \mathbb{E}\bigl[U\bigl(W_{T}(x_{0},\psi),B(\phi)\bigr)\bigr].
\end{equation}
As the $\bar{h}(B)_{t}$ are jointly Borel measurable and the $f_j$ are Borel measurable, we can show by induction that the $\bar{\psi}^{*}_{t}$ are  Borel functions, and thus $\psi^{*} \in \Phi$. \\
Fix $\psi\in\Phi$. We write $\psi_t=\bar{\psi}_t(\varepsilon^{t-1})$, where $\bar{\psi}_t$ is a  Borel function given by Doob\rq{}s theorem  for $1\leq t \leq T$. 
Notice that, by independence of $\varepsilon^{T}$ and $\hat{\varepsilon}^T$, and thus  of $\varepsilon^{T}$ and $B(\phi)=W_{T}(x_{0},\phi)(\hat{\varepsilon}^T)$,  we obtain that 

\begin{eqnarray*}
& & \hspace*{-1cm} \mathbb{E}\bigl[U\bigl(W_{T}(x_{0},\psi),B(\phi)\bigr)\bigr] =  \mathbb{E}\bigl[U\bigl(W_{T}(x_{0},\psi)(\varepsilon^{T}),W_{T}(x_{0},\phi)(\hat{\varepsilon}^T)\bigr)\bigr]\\
& = & 
\mathbb{E}\bigl[\mathbb{E}[U\bigl(W_{T}(x_{0},\psi)(\varepsilon^{T}),W_{T}(x_{0},\phi)(\hat{\varepsilon}^T)\bigr)\vert\mathcal{F}_{T}^{\varepsilon}]\bigr]\\
& = & 
\mathbb{E}\bigl[\mathbb{E}[U\bigl(W_{T}(x_{0},\psi)(e^{T}),W_{T}(x_{0},\phi)(\hat{\varepsilon}^T)\bigr)]\big\vert_{e^{T}=\varepsilon^{T}}\bigr]\\
&=& \mathbb{E}\bigl[V_{T}\bigl(B(\phi)\bigr)
\bigl(\varepsilon^{T},x_{0}+\sum_{j=1}^{T}\bar{\psi}_{j}(\varepsilon^{j-1})f_{j}(\varepsilon^{j})\bigr)\bigr]\\
&=& \mathbb{E}\Bigl[\mathbb{E}\bigl[V_{T}\bigl(B(\phi)\bigr)
\bigl(\varepsilon^{T},x_{0}+\sum_{j=1}^{T-1}\bar{\psi}_{j}(\varepsilon^{j-1})f_{j}(\varepsilon^{j})+
\bar{\psi}_{T}(\varepsilon^{T-1})
f_{T}(\varepsilon^{T})\bigr)\vert\mathcal{F}^{\varepsilon}_{T-1}\bigr]\Bigr]\\
&=& \mathbb{E}\Bigl[\mathbb{E}\bigl[V_{T}\bigl(B(\phi)\bigr)\bigl(e^{T-1},\varepsilon_T,x_{0}+\sum_{j=1}^{T-1}\bar{\psi}_{j}(e^{j-1})f_{j}(e^{j}) \\
 & & +
\bar{\psi}_{T}(e^{T-1})f_{T}(e^{T-1},\varepsilon_{T})\bigr)\bigr]\big\vert_{e^{T-1}=\varepsilon^{T-1}}\Bigr]\\
& \leq & \mathbb{E}\bigl[V_{T-1}\bigl(B(\phi)\bigr)
\bigl(\varepsilon^{T-1},x_{0}+\sum_{j=1}^{T-1}\bar{\psi}_{j}(\varepsilon^{j-1})f_{j}(\varepsilon^{j})\bigr)\bigr] \\
&=& \mathbb{E}\Bigl[\mathbb{E}\bigl[V_{T-1}\bigl(B(\phi)\bigr)
\bigl(\varepsilon^{T-1},x_{0}+\sum_{j=1}^{T-1}\bar{\psi}_{j}(\varepsilon^{j-1})f_{j}(\varepsilon^{j})\bigr)\vert\mathcal{F}^{\varepsilon}_{T-2}\bigr]\Bigr]\\
&\leq& \mathbb{E}\bigl[V_{T-2}\bigl(B(\phi)\bigr)
\bigl(\varepsilon^{T-2},x_{0}+\sum_{j=1}^{T-2}\bar{\psi}_{j}(\varepsilon^{j-1})f_{j}(\varepsilon^{j})\bigr)\bigr]=\ldots \leq V_{0}\bigl(B(\phi)\bigr)(x_{0}),	
\end{eqnarray*}
holds by repeated applications of Lemma \ref{condexp}: first we take $X_{1}=B(\phi)=W_{T}(x_{0},\phi)(\hat{\varepsilon}^T)$ and $X_{2}=\varepsilon^{T}$;
then $X_{1}=\varepsilon_{T}$ and $X_{2}=\varepsilon^{T-1}$, and so on. 
If we insert $\bar{\psi}=\bar{\psi}^{*}=\bar{\psi}^{*}(\phi)(\cdot,x_{0})$ in the above estimate then it holds with \emph{equalities} everywhere, i.e. 
$$ \mathbb{E}\bigl[U\bigl(W_{T}(x_{0},{\psi}^{*}(\phi)),B(\phi)\bigr)\bigr]=V_{0}\bigl(B(\phi)\bigr)(x_{0}),$$
and \eqref{optim} holds. Then,  
 taking the supremum over $\psi$ in \eqref{optim}
\begin{eqnarray*}
u(x_{0},\phi)=\sup_{\psi\in\Phi}\mathbb{E}\bigl[U\bigl(W_{T}(x_{0},\psi),B(\phi)\bigr)\bigr] \leq \mathbb{E}\bigl[U\bigl(W_{T}(x_{0},{\psi}^{*}(\phi)),B(\phi)\bigr)\bigr] \leq u(x_{0},\phi),
\end{eqnarray*}
as ${\psi}^{*}(\phi) \in \Phi.$ 
This implies that $$u(x_{0},\phi)=\mathbb{E}\bigl[U\bigl(W_{T}(x_{0},{\psi}^{*}(\phi)),B(\phi)\bigr)\bigr]= V_{0}\big(B(\phi)\big)(x_{0}).$$
The unicity of ${\psi}^{*}$ follows from the unicity of the $\hb_t$ for $1\leq t\leq T$. 
\end{proof}

\subsection{Fixed point theorem, and remaining proofs}

Recall that  ${\varepsilon}^{T-1}:=(\varepsilon_{1},\ldots,\varepsilon_{T-1})$. We now introduce $\mathcal{S}:=\mathrm{supp}(\varepsilon^{T-1})$, where $\mathrm{supp}(\cdot)$ refers to the support (see for example, p 441 of   \cite{AB}).  Theorems 12.7  and 12.14 of \cite{AB} show  that
$\mathbb{P}[\varepsilon^{T-1} \in . ]$ admits a unique support such that    $\mathbb{P}[\varepsilon^{T-1} \in \mathcal{S}]=1$, and 
\begin{align}
\label{defd1}
\mathrm{supp}(\varepsilon^{T-1}):=\bigcap \left\{ A \subset \mathbb{R}^{(T-1)\times m},\; \mbox{closed}, \; \mathbb{P}[\varepsilon^{T-1} \in A]=1\right\}.
\end{align}
Assumptions \ref{filtra}, \ref{price1}, \ref{una}, \ref{utility}, and \ref{mienk} will be in force in the rest of this section.
By independence of $(\varepsilon_{1},\ldots,\varepsilon_{T-1})$ under $\mathbb{P}$, 
$\mathcal{S}=\mathrm{supp}(\varepsilon_{1})\times\cdots\times\mathrm{supp}(\varepsilon_{T-1})$.

Let $C(\mathcal{S})$ denote the Banach space of $\mathbb{R}^{T}$-valued continuous functions on $\mathcal{S}$, equipped with the norm
$$
||\varphi||_{\infty}:=\sup_{e\in\mathcal{S}}|\varphi(e)|,\ \varphi\in C(\mathcal{S}).
$$

At this point, we explain an important identification. If $\phi\in\Phi$ then, by Doob's theorem, there are
Borel measurable functions $\bar{\varphi}_{t}:\mathbb{R}^{(t-1)\times m}\to\mathbb{R}$, $1 \leq t \leq T$ (we mean that $\bar{\varphi}_{1}$ is
a constant) such that 
$\phi_{t}=\bar{\varphi}_t(\varepsilon_{1},\ldots,\varepsilon_{t-1})$. Now let us define for all $1\leq t\leq T$ the functions $\tilde{\phi}_{t}:\mathcal{S}\to\mathbb{R}$
by setting 
$$\tilde{\phi}_{t}(e_{1},\ldots,e_{T-1}):=\bar{\varphi}_{t}(e_{1},\ldots,e_{t-1}).$$
In this way, we obtain a $\mathcal{B}(\mathcal{S})$-measurable function $\tilde{\phi}:=(\tilde{\phi}_1,\ldots, \tilde{\phi}_T)$ with 
$\tilde{\phi}:\mathcal{S}\to\mathbb{R}^{T}$ is such that the $t$th coordinate function $\tilde{\phi}_t$
depends uniquely on its first $t-1$ coordinates.
Conversely, if $\tilde{\phi}:\mathcal{S}\to\mathbb{R}^{T}$ is such a function, then definining 
$$\phi_{t}:=\tilde{\phi}_{t}(\varepsilon_{1},
\ldots,\varepsilon_{T-1}), \;1\leq t\leq T,$$ 
we obtain an element $\phi\in\Phi$. Indeed, each $\tilde{\phi}_{t}$ is $\mathcal{B}(\mathcal{S})$-measurable 
and as $\phi_{t}=\tilde{\phi}_{t}(\varepsilon_{1},
\ldots,\varepsilon_{t-1},0\ldots,0)$  $\phi_{t}$ is $\mathcal{F}^{\varepsilon}_{t-1}$-measurable. 
>From this moment on, we identify each $\phi\in\Phi${}
with a corresponding Borel measurable function $\tilde{\phi}:\mathcal{S}\to\mathbb{R}^{T}$. When we write $\phi\in C(\mathcal{S})$ we mean that the $\tilde{\phi}$
corresponding to $\phi$ can be chosen continuous.
Note also that
\begin{eqnarray*}
W_{t}(x_{0},\phi) &= & x_{0}+\sum_{j=1}^{t}\phi_{j}\Delta S_{j}=x_{0}+
\sum_{j=1}^{t}\tilde{\phi}_{j}(\varepsilon_{1},\ldots,\varepsilon_{T-1})f_{j}(\varepsilon_{1},\ldots,\varepsilon_{j})\\
B({\phi})&= &x_{0}+\sum_{t=1}^{T}\tilde{\phi}_{t}(\hat{\varepsilon}_{1},\ldots,\hat{\varepsilon}_{T-1})f_{t}(\hat{\varepsilon}_{1},\ldots,\hat{\varepsilon}_{t}),
\end{eqnarray*}
we stress one more time that here $\tilde{\phi}_{j}$ depends only on its first $j-1$ coordinates.
Finally, for each $M>0$, $\phi\in \Phi_{M}$ if and only if $\phi\in \Phi$, and for all $1 \leq t \leq T$, setting $\phi_{t}=\tilde{\varphi}_t(\varepsilon_{1},\ldots,\varepsilon_{T-1})$ as before, 
for all $e^{T-1},\bar{e}^{T-1}\in\mathbb{R}^{(T-1)\times m}$,
\begin{eqnarray}
\label{defcm}
|\tilde{\varphi}_{t}(e^{T-1})|\leq M \mbox{  and  } |\tilde{\varphi}_{t}(e^{T-1})-\tilde{\varphi}_{t}(\bar{e}^{T-1})|  \leq   M|e^{T-1}-\bar{e}^{T-1}|^{\chi/2^{T-t+1}}.
\end{eqnarray}
It is clear that $\Phi_{M} \subset C(\mathcal{S})$.  Moreover,  Proposition \ref{arzela-ascoli} below shows that $\Phi_{M}$ is relatively compact in 
$C(\mathcal{S})$. Indeed, the left-hand side of \eqref{defcm} implies that $|\tilde{\varphi}|\leq M \sqrt{T}$, which proves  
the first condition of  Proposition \ref{arzela-ascoli}. For the second one, let $\phi \in  \Phi_{M}$, $e^{T-1},\bar{e}^{T-1}\in\mathbb{R}^{(T-1)\times m}$
\begin{eqnarray*}
  |\tilde{\varphi}(e^{T-1})-\tilde{\varphi}(\bar{e}^{T-1})| &=&  \Big(\sum_{t=1}^T|\tilde{\varphi}_{t}(e^{T-1})-\tilde{\varphi}_{t}(\bar{e}^{T-1})|^2\Big)^{1/2} \\
  & \leq &   M\Big(\sum_{t=1}^T|e^{T-1}-\bar{e}^{T-1}|^{\chi/2^{T-t}}\Big)^{1/2} \\
  & \leq & \sqrt{T M^2 +2T M^2}|e^{T-1}-\bar{e}^{T-1}|^{\chi/2^{T-2}},
\end{eqnarray*}
reasonning as in Lemma \ref{upi}, which shows \eqref{equi}. 
Moreover, $\Phi_{M}$ is trivially closed, and thus compact. 

One key result for our arguments is the following.

\begin{proposition}\label{contin}
Let Assumptions \ref{filtra}, \ref{price1}, \ref{una}, \ref{utility} and \ref{mienk} hold. 
Let $x_{0} \in \mathbb{R}$. For all $\phi\in\Phi$, let  $\psi^{*}:=\psi^{*}(\phi)(\cdot,x_{0})\in\Phi$  
be the optimizer  of \eqref{petitu} given by Proposition \ref{dyna}. 
Then, the mapping $\phi \mapsto \psi^{*}(\phi)$ is continuous (for the norm of $C(\mathcal{S})$) from 
$\Phi_{C(x_{0})}$ to $\Phi_{C(x_{0})}$. 
\end{proposition}
\begin{proof} 
Recall the notation of Proposition \ref{dyna}. However, for ease of exposition, we don\rq{}t indicate the dependence of $\psi^{*}$ on $x_{0}$. 
We make the identification above and associate to  $\psi^{*}(\phi)$, the function $\tilde{\psi}^{*}(\phi):\mathcal{S}\to\mathbb{R}^{T}$, i.e. $\tilde{\psi}^{*}(\phi)_t(e^{T-1})=\bar{\psi}^{*}(\phi)_{t}(e^{t-1}).$  
Using \eqref{m1} and \eqref{m2},  $\psi^{*}(\phi) \in \Phi_{C(x_{0})}.$ So, $\phi \mapsto \tilde{\psi}^{*}(\phi)$ maps $\Phi_{C(x_{0})}$ (in fact, the whole of $\Phi$) 
into $\Phi_{C(x_{0})}$. 

Now let $(\tilde{\phi}^{n})_n\subset\Phi_{C(x_{0})}$ that converge to $\tilde{\phi}$ in the topology of the Banach space $C(\mathcal{S})$, i.e. $\Vert \tilde{\phi}^{n}-\tilde{\phi}\Vert_{\infty}\to 0, \, n\to \infty$. 
We call ${\phi}^{n}$ and ${\phi}$ the associated elements of $\Phi$.  
We want to prove that $\Vert \tilde{\psi}^{*}({\phi}^{n})-\tilde{\psi}^{*}(\phi)\Vert_{\infty}\to 0, \, n\to \infty$. 

First, remark that for all $\omega \in\Omega$ and $n\in \mathbb{N}$ 
\begin{eqnarray}
|B({\phi}^{n})(\omega)-B({\phi})(\omega)|\leq \sum_{j=1}^{T}\Vert \tilde{\phi}^{n}-\tilde{\phi}\Vert_{\infty}|f_{j}\big(\hat{\varepsilon}^{j}(\omega)\big)|\leq TC_{f}
\Vert \tilde{\phi}^{n}-\tilde{\phi}\Vert_{\infty}\label{wpo}	
\end{eqnarray}
so $B({\phi}^{n})(\omega)\to B(\phi)(\omega)$ for all $\omega\in\Omega$. 
Define the random utility functions for all $n \in \mathbb{N}$, $x \in  \mathbb{R}$, $\omega \in \Omega$
\begin{eqnarray*} 
\mathfrak{U}_{n}(\omega,x) & := & U\big(x,B({\phi}^{n})(\omega)\big)=U(x)+\nu\Big(U(x)-U\big (B({\phi}^{n})(\omega) \big)\Big)  \\
 \mathfrak{U}_{\infty}(\omega,x) & := & U\big(x,B({\phi})(\omega)\big)=U(x)+\nu\Big(U(x)-U\big (B({\phi})(\omega) \big)\Big). 
 \end{eqnarray*}  
 Let $\bar{\mathbb{N}}:=\mathbb{N} \cup \{\infty\}$. 
We will now verify the conditions of Theorem \ref{util} below. Assumptions \ref{utility} and \ref{mienk} imply that  for all $n\in \bar{\mathbb{N}}$, each $\mathfrak{U}_{n}$ is  strictly concave and increasing, continuously differentiable in $x$.
Using \eqref{nuleibnitz}, for all $x_1,x_2 \in  \mathbb{R},$ all random variables $B_1,B_2$
\begin{eqnarray} 
|U(x_1,B_1)-U(x_2,B_2)| & \leq & (1+k_{-})|U(x_1)-U(x_2)| + k_{-}|B_1 -B_2|. 
\end{eqnarray} 
This implies that 
\begin{eqnarray} 
\nonumber
\sup_{n \in \bar{\mathbb{N}},\omega\in\Omega}\bigl[\mathfrak{U}_{n}(\omega,\infty)-\mathfrak{U}_{n}(\omega,x)\bigr] & = & 
\sup_{n \in \bar{\mathbb{N}},\omega\in\Omega} |U\big(\infty,B({\phi}^{n})(\omega)\big)-U\big(x,B({\phi}^{n})(\omega)\big)|\\
 & \leq & (1+k_{-})|U(\infty)-U(x)| \to 0,\label{santafe}
\end{eqnarray} 
as $x\to\infty$, implying \eqref{egyutt} below, and also
\begin{eqnarray*} 
\bigl|\mathfrak{U}_{n}(\omega,x)-\mathfrak{U}_{\infty}(\omega,x)\bigr|  &= &  \bigl|U\big(x,B({\phi}^{n})(\omega)\big)-U\big(x,B({\phi})(\omega)\big)\bigr| \\
& \leq &  k_{-}|B({\phi}^{n})(\omega) -B({\phi})(\omega)|,
\end{eqnarray*} 
and \eqref{wpo}	 implies that $\mathfrak{U}_{n}(\omega,x) \to  \mathfrak{U}_{\infty}(\omega,x),$ $n\to \infty$ for all $\omega \in \Omega$ 
and for all $ x \in  \mathbb{R}.$  So, \eqref{conv} holds true. 
Moreover,  
\begin{eqnarray*}
\mathfrak{U}_{n}'(\omega,0) =U'\big(0,B({\phi}^{n})(\omega)\big),  \quad 
\mathfrak{U}_{\infty}'(\omega,0) =U'\big(0,B({\phi})(\omega)\big)
\end{eqnarray*}
so \eqref{borneinfUprime} implies that  
$\mathrm{ess.}\inf_{n\in\bar{\mathbb{N}}}\mathfrak{U}_{n}'(\omega,0)\geq U'(0)>0$ and \eqref{zero} holds true. 
Morover, \eqref{borneinfU} implies that 
\begin{eqnarray*}
(1+k_{-})U(x) -k_{-}C_{U} \leq U(x,B) \leq C_U+ C_{\nu}. 
\end{eqnarray*}
So,  for all $x\in\mathbb{R}$ and $\omega \in \Omega$ 
  \begin{eqnarray*}
  \mathrm{ess.}\sup_{n\in\bar{\mathbb{N}}}|\mathfrak{U}_{n}(\omega,x)|& \leq  & (1+k_{-})C_U+ C_{\nu} + (1+k_{-})|U(x)|<\infty,
  \end{eqnarray*}
As, for all $n\in\bar{\mathbb{N}}$ and $\omega \in \Omega$ , $\mathfrak{U}_{n}(\omega,\infty) \geq \mathfrak{U}_{n}(\omega,0)$, 
\begin{eqnarray*}
C_U+ C_{\nu} \geq \mathrm{ess.}\inf_{n\in\bar{\mathbb{N}}}\mathfrak{U}_{n}(\omega,\infty) \geq   (1+k_{-})U(0) -k_{-}C_{U} >-\infty,
  \end{eqnarray*}
and \eqref{fini} holds true. 

Recall that 
${\psi}^{*}(\phi^{n})$ (resp. ${\psi}^{*}(\phi)$) is the optimizer of \eqref{petitu} given by Proposition \ref{dyna} for $\phi^{n}$ 
(resp. $\phi$). It is thus the (unique) optimizer for $\mathfrak{U}_{n}$ (resp. $\mathfrak{U}_{\infty}$) in \eqref{optil}. 
Theorem \ref{util} shows that
\begin{equation}\label{coco}
\psi^{*}(\phi^{n})_{t}\to\psi^{*}(\phi)_{t}
\end{equation}
almost surely for all $1 \leq t \leq T$. 
Recalling the notation of the beginning of the proof, providing the identification of the above of strategies with continuous functions on $\mathcal{S}$, for all $n\in\mathbb{N}$, we have
$\psi^{*}(\phi^{n})=\tilde{\psi}^{*}(n)(\varepsilon^{T-1})$ and $\psi^{*}(\phi)=\tilde{\psi}^{*}(\varepsilon^{T-1}).$ For ease of notation, we set $\tilde{\psi}^{*}(n)$ for $\psi^{*}(\phi^{n})$  and $\tilde{\psi}^{*}$ for $\psi^{*}(\phi)$. 
 Then, \eqref{coco} implies that for all $1 \leq t \leq T$,
\begin{equation}\label{varo}
\tilde{\psi}^{*}_{t}(n)\to\tilde{\psi}^{*}_{t} 
\end{equation}
$\mu$-almost surely, where $\mu$ denotes the law of $\varepsilon^{T-1}$ under $\mathbb{P}$. Since $\mathcal{S}$ is the support of $\mu$, 
$\tilde{\psi}^{*}_{t}(n)\to\tilde{\psi}^{*}_{t}$ pointwise on a dense  subset of $\mathcal{S}$, see Lemma \ref{dense} below.

Seeking a contradiction, suppose that $\tilde{\psi}^{*}({n})$ do not converge to 
$\tilde{\psi}^{*}$ in the norm of $C(\mathcal{S})$. Then, along a subsequence (still denoted by $n$)
we have 
\begin{equation}\label{happy}
\inf_{n}||\tilde{\psi}^{*}(n)-\tilde{\psi}^{*}||_{\infty}>0.	
\end{equation} 
By compactness of $\Phi_{C(x_{0})}$, a further subsequence of $\tilde{\psi}^{*}({n})$ can be chosen (still denoted by $n$)
such that $||\tilde{\psi}^{*}(n)-\hat{\psi}||_{\infty}\to 0$, $n\to\infty$ for some $\hat{\psi} \in \Phi_{C(x_{0})}$. 
Since \eqref{varo} holds pointwise on a dense subset of $\mathcal{S}$, $\hat{\psi}=\tilde{\psi}^{*}$ on this set and, by continuity, on the whole of $\mathcal{S}$.
But this contradicts \eqref{happy}.
\end{proof}

We can finally achieve the proof of our main result.

\begin{proof}[Proof of Theorem \ref{main}]
Recalling the notation of Proposition \ref{contin}, the mapping $\phi \mapsto \psi^{*}(\phi)$ from $\Phi_{C(x_{0})}$ to $\Phi_{C(x_{0})}$ is continuous for the norm of $C(\mathcal{S})$. 
The set $\Phi_{C(x_{0})}$ is compact in 
$C(\mathcal{S})$ and also trivially convex. 
With the choice $\mathbb{B}:=C(\mathcal{S})$ and $H=\Phi_{C(x_{0})}$, 
Theorem \ref{schauder} below gives a fixed point, i.e. some $\phi^{\dagger} \in \Phi_{C(x_{0})}$ such that $\phi^{\dagger}= \psi^{*}(\phi^{\dagger}).$ 
This implies that, 
$$
u(x_{0},\phi^{\dagger})=\sup_{\psi\in\Phi}\mathbb{E}\bigl[U\bigl(W_{T}(x_{0},\psi),B(\phi^{\dagger})\bigr)\bigr]=\mathbb{E}\bigl[U\bigl(W_{T}(x_{0},\phi^{\dagger}),B(\phi^{\dagger})\bigr)\bigr],
$$	
and $\phi^{\dagger}$ is by definition a personal equilibrium. 

We now prove the existence of a preferred equilibrium. 
It is convenient to introduce the notation, for $\phi,\psi\in \Phi$,
$$
\mathcal{U}(\phi,\psi):=\mathbb{E}\bigl[U\bigl(W_{T}(x_{0},\phi),B(\psi)\bigr)\bigr].
$$

Now let $\phi^{\dagger}(n)\in\Phi^{\dagger}$ be a sequence such that 
$$\mathcal{U}\big(\phi^{\dagger}(n),\phi^{\dagger}(n)\big)\to 
\sup_{\phi\in\Phi^{\dagger}}\mathcal{U}(\phi,\phi),\ n\to\infty.$$ 
By compactness of $\Phi_{C(x_{0})}$, there is a subsequence (still denoted by $n$) 
and $\phi^{\sharp}\in\Phi_{C(x_{0})}$ such that $\phi^{\dagger}(n)\to \phi^{\sharp}$ in the topology of $C(\mathcal{S})$. 
In particular, an estimate like \eqref{wpo} shows that for all $\omega \in \Omega$,  $B(\phi^{\dagger}(n))(\omega)\to B(\phi^{\sharp})(\omega)$ and also
$$
U\Bigl(W_{T}\big(x_{0},\phi^{\dagger}(n)\big)(\omega),B\big(\phi^{\dagger}(n)\big)(\omega)\Bigr)\to 
U\bigl(W_{T}(x_{0},\phi^{\sharp})(\omega),B(\phi^{\sharp})(\omega)\bigr).
$$ 
Dominated convergence implies $
\mathcal{U}(\phi^{\dagger}(n),\phi^{\dagger}(n)) \to \mathcal{U}(\phi^{\sharp},\phi^{\sharp})$ and 
$$
\mathcal{U}(\phi^{\sharp},\phi^{\sharp})= 
\sup_{\phi\in\Phi^{\dagger}}\mathcal{U}(\phi,\phi).$$ 

It remains to show that $\phi^{\sharp}$ itself is a personal equilibrium, i.e. $u(x_{0},\phi^{\sharp})=\mathcal{U}(\phi^{\sharp},\phi^{\sharp})$. 
By Proposition \ref{dyna}, there is an optimizer $\psi^*=\psi^*(\phi^{\sharp})\in \Phi_{C(x_{0})}$ such that
$$u(x_{0},\phi^{\sharp})=\mathbb{E}\bigl[U(W_{T}\bigl(x_{0},\psi^*),B(\phi^{\sharp})\bigr)\bigr]=\mathcal{U}(\psi^*,\phi^{\sharp}).$$ 
(At this point, we do not know yet that $\psi^*=\phi^{\sharp}$.) 
Since $\phi^{\dagger}(n)$ was
a personal equilibrium, for all $n$,
$$
\mathcal{U}\big(\phi^{\dagger}(n),\phi^{\dagger}(n)\big)\geq \mathcal{U}\big(\psi^*,\phi^{\dagger}(n)\big).
$$
Passing to the limit (again by dominated convergence),
$$
\mathcal{U}(\phi^{\sharp},\phi^{\sharp})\geq \mathcal{U}(\psi^*,\phi^{\sharp})=u(x_{0},\phi^{\sharp})\geq \mathcal{U}(\psi,\phi^{\sharp})
$$ 
for all $\psi\in\Phi$. Choosing $\psi=\phi^{\sharp}$, we have equality,  
so $\phi^{\sharp}$ is indeed a personal equilibrium, and we may conclude. 
We remark that, by uniqueness of the optimizer, necessarily $\phi^{\sharp}=\psi^*$. 
\end{proof}

\section{Auxiliary results}\label{sec4}

\begin{proposition}
\label{lipcomp}
Let $K_t \subset \mathbb{R}^{t\times m}$ be a non-empty compact set. Let $C_f>0$ and let $\chi\in(0,1]$. 
Let $f_t : K_t \to \mathbb{R}$ such that
\begin{eqnarray}
\label{f1}
|f_t(e^t)-f_t(\bar e^t)| & \le &  C_f |e^t-\bar e^t|^\chi,
\qquad\forall e^t,\bar e^t\in K_t,\\
\label{f2}
|f_t(e^t)| & \le & C_f \qquad\forall e^t \in K_t.
\end{eqnarray}
Define $F_t,g_t:\mathbb{R}^{t\times m}\to\mathbb{R}$ by
\begin{eqnarray*}
F_t(e^t) &:= & \inf_{\bar e^t \in K_t} 
\bigl( f_t(\bar e^t) + C_f \, |e^t - \bar e^t|^{\chi} \bigr), \qquad\forall e^t \in\mathbb{R}^{t\times m}.\\
g_t(e^t) & := &  
\begin{cases}
F_t(e^t), & |e^t| \le R,\\[2mm]
F_t(\pi_R(e^t)), & |e^t| > R,
\end{cases}
\end{eqnarray*}
where $R>0$ is such that $K_t \subset B(0,R)$, and $\pi_R(e^t)$ denotes the projection of $e^t$ onto the 
closed ball of $\mathbb{R}^{t\times m}$ of centre $0$ and radius $R$, $B(0,R)$. Then, $g_t|_{K_t}=f_t$ and 
\begin{eqnarray}
\label{g1}
|g_t(e^t)-g_t(\bar e^t)|
&\le & C_f\,|e^t-\bar e^t|^\chi,
\qquad\forall e^t,\bar e^t\in\mathbb{R}^{t\times m} \\
\label{g2}
|g_t(e^t)| & \le & C_f(1+(2R)^\chi), \qquad\forall e^t \in K_t.
\end{eqnarray}
If $K_t=B(0,R)$, then \eqref{g2} holds with $C_f$ instead of $C_f(1+(2R)^\chi) $. 
\end{proposition}

\begin{proof}
We first prove that 
\begin{eqnarray}
\label{F1}
|F_t(e^t)-F_t(\bar e^t)|
\le C_f\,|e^t-\bar e^t|^\chi,
\qquad\forall e^t,\bar e^t\in\mathbb{R}^{t\times m}.
\end{eqnarray}
Fix $e^t\in\mathbb{R}^{t\times m}$.  Since $f_t$ satisfies \eqref{f1} for all $\bar e^t, \bar z^t\in K_t$
\begin{eqnarray*}
F_t(e^t) &\leq  & f_t(\bar e^t)+C_f\,|e^t-\bar e^t|^\chi\\
& \leq  & 
 f_t(\bar z^t)+C_f|\bar e^t-\bar z^t|^\chi +C_f |e^t-\bar e^t|^\chi.
\end{eqnarray*}
So, taking the infimum over $\bar z^t\in K_t$, we get that $F_t(e^t)\le F_t(\bar e^t)+C_f\,|e^t-\bar e^t|^\chi$. Then, \eqref{F1} follows by symmetry. \\
\noindent\textit{Extension property.}\\
If $e^t\in K_t\subset B(0,R)$, then  $g_t(e^t)=F_t(e^t)=f_t(e^t)$.
Thus $g_t|_{K_t}=f_t$.\\
\noindent \textit{Hölder property.}\\
Case 1: $e^t, \bar e^t \in B(0,R)$.
Then, \eqref{F1} shows that 
$$
|g_t(e^t) - g_t(\bar e^t)| = |F_t(e^t)-F_t(\bar e^t)| \le C_f |e^t - \bar e^t|^\chi.
$$
{Case 2: $e^t, \bar e^t \notin B(0,R)$.}  
Since the projection $\pi_R$ is $1$-Lipschitz (see  \cite[Lemma 6.54]{AB}),   \eqref{F1} again shows that
$$
|g_t(e^t) - g_t(\bar e^t)| 
= |F_t(\pi_R(e^t)) - F_t(\pi_R(\bar e^t))| 
\le C_f |\pi_R(e^t)-\pi_R(\bar e^t)|^\chi 
\le C_f |e^t-\bar e^t|^\chi.
$$
{Case 3: one point inside $B(0,R)$, one outside.}  
Without loss, assume $|e^t| \le R < |\bar e^t|$. 
Let $\underline e^t= (1-s_0)e^t+s_0\bar e^t$ be the intersection of $\{(1-s)e^t+s\bar e^t: s\in [0,1]\}$ with the sphere $\partial B(0,R)$. Note that $\pi_R(\underline{e}^t)=\underline{e}^t$. 
Using \eqref{F1} again, we get that 
\begin{eqnarray*}
|F_t(e^t)-F_t(\underline e^t)| & \le &  C_f |e^t-\underline e^t|^{\chi}=s_0C_f |e^t-\bar e^t|^{\chi} \\ 
|F_t(\underline e^t)-F_t\big(\pi_R(\bar e^t)\big)| &= & |F_t\big(\pi_R(\underline{e}^t)\big)-F_t\big(\pi_R(\bar e^t)\big)|
 \le  C_f|\pi_R(\underline{e}^t)-\pi_R(\bar e^t)|^{\chi}\\
& \leq & C_f|\underline e^t-\bar e^t|^{\chi}=(1-s_0)C_f |e^t-\bar e^t|^{\chi}\\
|g_t(e^t)-g_t(\bar e^t)| &= & |F_t(e^t)-F_t\big(\pi_R(\bar e^t)\big)| \\
& \le &  |F_t(e^t)-F_t(\underline e^t)| 
 + |F_t(\underline e^t)-F_t\big(\pi_R(\bar e^t)\big)|
 \le  C_f |e^t-\bar e^t|^{\chi}. 
\end{eqnarray*}
\noindent  \textit{Uniform boundedness.} \\
If $ e_t \in B(0,R)$, then 
$$
g_t(e^t)=F_t(e^t)
\le C_f + C_f \inf_{\bar e^t\in K_t} |e^t-\bar e^t|^\chi \le C_f(1+(2R)^\chi),$$
as $K_t  \subset B(0,R)$. Note that if $K_t  = B(0,R)$, $\inf_{\bar e^t\in K_t} |e^t-\bar e^t|^\chi=0$. 
If $ e_t\notin B(0,R)$, then $g_t(e^t)=F_t(\pi_R(e^t))$, so the preceding inequality applies as $ \pi_R(e^t) \in B(0,R)$.
\end{proof}

\begin{lemma}\label{vege}
The model of Example \ref{eex} satisfies Assumptions \ref{price1} and \ref{una}.	
\end{lemma}
\begin{proof}
Let us fix $C_{\varepsilon}\geq 1$ such that $|\varepsilon_{t+1}|\leq C_{\varepsilon}$.
Let $f_{t}(e^{t}):=\mu_{t}(e^{t-1})+\sigma_{t}(e^{t-1})e_{t}$. Trivially, $|f_{t}|\leq C+C C_{\varepsilon}$ on $[- C_{\varepsilon}, C_{\varepsilon}]^{t}$. 
For $e^{t}, \bar{e}^{t} \in [- C_{\varepsilon}, C_{\varepsilon}]^{t}$, 
\begin{eqnarray*}
& & \hspace*{-1cm} |f_t(e^t) -f_t(\bar{e}^t)| \\
&\leq& |\mu_{t}(e^{t-1})-\mu_{t}(\bar{e}^{t-1})|+|\sigma_{t}(e^{t-1})-\sigma_{t}(\bar{e}^{t-1})||e_{t}|+
|e_{t}-\bar{e}_{t}||\sigma_{t}(\bar{e}^{t-1})|\\
&\leq& C|e^{t-1}-\bar{e}^{t-1}|^{\delta}+C|e^{t-1}-\bar{e}^{t-1}|^{\delta}C_{\varepsilon}+ 
C|e_{t}-\bar{e}_{t}|\\
&\leq & [3(C+ C C_{\varepsilon})+2(C+C C_{\varepsilon})]|e^{t}-\bar{e}^{t}|^{\delta}= 5C(1+ C_{\varepsilon})|e^{t}-\bar{e}^{t}|^{\delta},
\end{eqnarray*}
using Lemma \ref{upi} below. Now, define $g_t$ as in Proposition \ref{lipcomp} with $R=\sqrt{T}C_{\varepsilon}$. Clearly, 
$\Delta S_{t}=g_{t}(\varepsilon_{1},\ldots,\varepsilon_{t})$  and Assumption \ref{price1} holds for $g_t$. It remains to check Assumption \ref{una}.
By our hypothesis on $\varepsilon_{t+1}$,
\begin{eqnarray*}
\mathbb{P}[\mu_{t+1}(e^{t})+\sigma_{t+1}(e^{t})\varepsilon_{t+1}\leq -\beta] \geq \mathbb{P}\left[\varepsilon_{t+1}\leq \frac{-C-\beta}{c}\right]
\geq \beta \\
\mathbb{P}[\mu_{t+1}(e^{t})+\sigma_{t+1}(e^{t})\varepsilon_{t+1}\geq \beta] \geq 
\mathbb{P}\left[\varepsilon_{t+1}\geq \frac{C+\beta}{c}\right]\geq \beta
\end{eqnarray*}
and we choose $\alpha=\beta$.

\end{proof}

Simple observations are noted next. 

\begin{proof}[Proof of Lemma \ref{epsilontilde}]
There is a
bijection $\zeta:\mathbb{R}\to\mathbb{R}^{T\times m}$ such that
$\zeta,\zeta^{-1}$ are Borel measurable; see \cite[Corollary 7.16.1, p.122]{BS}. 

Consider the probability $\kappa(A):=\mathbb{P}[(\varepsilon_{1},\ldots,\varepsilon_{T})\in \zeta(A)]$,
defined for $A\in\mathcal{B}(\mathbb{R})$. The corresponding cumulative distribution function is
$$
F_{\kappa}(x):=\kappa[(-\infty,x]],\ x\in \mathbb{R},
$$ 
and its pseudo-inverse is
$$
F_{\kappa}^{-}(u):=\inf\{x:F_{\kappa}(x)\geq u\},\ u\in (0,1).
$$
It is well-known that the random variable $F_{\kappa}^{-}(\hat{\varepsilon})$ has law $\kappa$ under $\mathbb{P}$.{}
Define $\Upsilon(u):=\zeta(F_{\kappa}^{-}(u))$. By the definition of $\kappa$, $\Upsilon(\hat{\varepsilon})$
has the same law as $(\varepsilon_{1},\ldots,\varepsilon_{T})$.
\end{proof}

\begin{lemma}\label{upi} Let $n,N,M\in\mathbb{N}$ and $0\leq \theta_1\leq \theta_2 \leq \ldots \leq  \theta_n$. 
Let $f:\mathbb{R}^{N}\to\mathbb{R}^M$ be a function with $|f|\leq \bar{C}$. 
If, for all 
$e,\bar{e}\in \mathbb{R}^{N}$,
\begin{equation}\label{flat}
|f(e)-f(\bar{e})|\leq C \sum_{i=1}^n|e-\bar{e}|^{\theta_i}
\end{equation}
for some constant $C>0$, then 
\begin{equation}\label{natural}
|f(e)-f(\bar{e})|\leq 
[nC+2\bar{C}]|e-\bar{e}|^{\theta_1},\ e,\bar{e}\in \mathbb{R}^{N}.
\end{equation}
\end{lemma}
\begin{proof} 
If $|e-\bar{e}|<1$ then \eqref{flat} implies $|f(e)-f(\bar{e})|\leq nC|e-\bar{e}|^{\theta_1}.$ 
In the opposite case,
\begin{eqnarray*}
 |f(e)-f(\bar{e})| \leq 2\bar{C}\leq 2\bar{C}|e-\bar{e}|^{\theta_1},
\end{eqnarray*}	
showing our claim.
\end{proof}

\begin{lemma}\label{cont} Let $F : \mathbb{R} \to \mathbb{R} $ and $K : \mathbb{R} \to (0, +\infty) $ be two continuous functions. Define
$$
f(x) := \sup_{y \in [x - K(x), x + K(x)]} F(y).
$$
Then $f $ is continuous on $\mathbb{R}$.
\end{lemma}
\begin{proof} 
Fix $ x \in \mathbb{R}$. We perform the change of variable 
$
t := {(y - x)}/{K(x)}, 
$
and  rewrite:
$$
f(x) = \sup_{t \in [-1, 1]} F(x + K(x)t).
$$
Now define the auxiliary function on $\mathbb{R} \times [-1, 1]$ by $\varphi(x, t) := F(x + K(x)t).$ 
Since $ F$ and $K$ are continuous, so is $\varphi$, and we can apply the Maximum Theorem (\cite[Theorem 17.31]{AB}) on the 
compact $[-1,1]$: $f$ is continous. 
\end{proof}


The following measure-theoretical result was needed in our argument for dynamic programming above.

\begin{lemma}\label{condexp}
Let $X_{1}\in\mathbb{R}^{d_{1}}$, $X_{2}\in\mathbb{R}^{d_{2}}$ be independent random variables, 
$\Xi:\mathbb{R}^{d_{1}}\times\mathbb{R}^{d_{2}}\to\mathbb{R}$ be Borel 
measurable and bounded from above. Define 
$$
\Xi^{\sharp}(x_{2}):=\mathbb{E}[\Xi(X_{1},x_{2})],\ x_{2}\in\mathbb{R}^{d_{2}}.
$$ 
Then $\Xi^{\sharp}(X_{2})$ is a version of $\mathbb{E}[\Xi(X_{1},X_{2})\vert\sigma(X_{2})]$. We may write
$$\Xi^{\sharp}(X_{2})=\mathbb{E}[\Xi(X_{1},x_{2})]|_{x_2=X_2}.$$ 
\end{lemma}
\begin{proof}
By standard measure-theoretic arguments, we can reduce that statement to the case where $\Xi(x_{1},x_{2})=1_{A_{1}}(x_{1})1_{A_{2}}(x_{2})$
with Borel sets $A_{1},A_{2}$. By independence, we get that $\mathbb{P}$-a.s.
\begin{eqnarray*}
\mathbb{E}[\Xi(X_{1},X_{2})\vert\sigma(X_{2})] &= & 1_{A_{2}}(X_{2})\mathbb{E}[1_{A_{1}}(X_{1})\vert\sigma(X_{2})]\\
 &= & 1_{A_{2}}(X_{2})\mathbb{E}[1_{A_{1}}(X_{1})] = \Xi^{\sharp}(X_{2}),
\end{eqnarray*}	
finishing the proof. 
\end{proof}

\begin{lemma}\label{dense}
Let $k \geq 1$. Assume that $f_n \to f$, $n\to \infty$ $\mu$-a.s. where $\mu$ is a probability measure on $\mathcal{B}(\mathbb{R}^k)$. 
Then, there exists a dense subset $D$ of 
$\mathrm{supp}(\mu)$, where $f_n(x) \to f(x)$, $n\to \infty$, for all $x\in D$. 
\end{lemma}
\begin{proof}
Let $A\in \mathcal{B}(\mathbb{R}^k)$ such that for all $x\in A$, $f_n(x) \to f(x)$, $n\to \infty$ and $\mu[A]=1$. 
Notice that $\mu[\mathrm{supp}(\mu)]=1$ hence $D:=A\cap\mathrm{supp}(\mu)$ and its closure $\bar{D}$ also
satisfy $\mu[D]=\mu[\bar{D}]=1$. Then $\bar{D}\supset\mathrm{supp}(\mu)$ since the latter
is the smallest closed set of full $\mu$-measure, see \eqref{defd1}. This means precisely that $D$ is dense in $\mathrm{supp}(\mu)$. 
\end{proof}




We recall the main result of \cite{cr2007b} in a form that is convenient for the present setting.

\begin{theorem}\label{util} Let Assumptions \ref{filtra}, \ref{price1}, \ref{una} be in vigour.
For all $n \in \bar{\mathbb{N}}:=\mathbb{N} \cup \{\infty\}$,  let the random utilities 
$\mathfrak{U}_{n}:\Omega\times \mathbb{R}\to\mathbb{R}$ satisfy 
\begin{eqnarray}
\label{fini}
-\infty<\mathrm{ess.}\inf_{n\in\bar{\mathbb{N}}}\mathfrak{U}_{n}(\cdot,\infty) <+\infty  \mbox{ a.s.}  \quad \mathrm{ess.}\sup_{n\in\bar{\mathbb{N}}}|\mathfrak{U}_{n}(\cdot,x)|<\infty \mbox{ a.s. } \forall x\in\mathbb{R}.& &\\
\label{egyutt}
\lim_{x\to\infty}\sup_{n\in\bar{\mathbb{N}},\omega\in\Omega}[\mathfrak{U}_{n}(\omega,\infty)-\mathfrak{U}_{n}(\omega,x)]= 0 & &
\end{eqnarray}
Assume that each $\mathfrak{U}_{n}$ is (almost surely) strictly concave and increasing, continuously differentiable in $x$, 
with \begin{eqnarray} 
\label{zero}\mathrm{ess.}\inf_{n\in\bar{\mathbb{N}}}\mathfrak{U}_{n}'(\cdot,0)>0 \mbox{ a.s.}
\end{eqnarray}
Furthermore, assume that for each $x\in\mathbb{R}$ 
\begin{eqnarray} 
\label{conv}
\mathfrak{U}_{n}(\cdot,x)\to \mathfrak{U}_{\infty}(\cdot,x) \mbox{ a.s.,} n\to\infty.
\end{eqnarray}
Let $x_{0} \in\mathbb{R}.$ Then, for all $n \in \bar{\mathbb{N}}$, there are (a.s.) unique optimizers $\Psi(n):=\Psi(n)(\cdot,x_{0})$, $\Psi(n)\in \Phi$ satisfying
\begin{eqnarray}
\label{optil}
\mathbb{E}\Bigl[\mathfrak{U}_{n}\Bigl(\cdot, W_{T}\bigl(x_{0},\Psi(n)\bigr)\Bigr)\Bigr]={}
\sup_{\phi\in\Phi}\mathbb{E}\Bigl[\mathfrak{U}_{n}\bigl(\cdot, W_{T}(x_{0},\phi)\bigr)\Bigr].
\end{eqnarray}
and  for all $1\leq t\leq T$
\begin{equation}\label{connie}
\Psi(n)_{t}(\cdot,x_{0})\to \Psi({\infty})_{t}(\cdot,x_{0})  \mbox{ a.s. }, n\to\infty.
\end{equation}
Moreover, 
$$\lim_{n \to \infty}\mathbb{E}\Bigl[\mathfrak{U}_{n}\Bigl(\cdot, W_{T}\bigl(x_{0},\Psi(n)\bigr)\Bigr)\Bigr]=\mathbb{E}\Bigl[\mathfrak{U}_{{\infty}}\Bigl(\cdot, W_{T}\bigl(x_{0},\Psi({\infty})\bigr)\Bigr)\Bigr],$$
uniformly on compact sets. 
\end{theorem}
\begin{proof} 
Hypothesis (R) of \cite{cr2007b} is
automatic from Assumption \ref{una}. Indeed, let $\mu$ denote the law of $f_{t}(e^{t-1},\varepsilon_{t})$ under 
$\mathbb{P}$. 
There exists $x_+\in\mathrm{Supp}(\mu)$ with $x_+\ge\alpha$. Else 
$[\alpha,+\infty)$ would be disjoint from the support and therefore
$
\mu\big[[\alpha,+\infty)\big]=0,
$
contradicting $\alpha>0$ in \eqref{ep}.  Similarly there exists $x_-\in\mathrm{Supp}(\mu)$ with $x_-\le-\alpha$. 
As  the support contains at least two distinct points, 
the affine hull 
is the whole real line.  

By Remark \ref{conpr}, Assumption 2.1 of \cite{cr2007b} also holds. 
Let $\bar{ \Omega}$ be the full measure set where all the assumptions of Theorem \ref{util} hold.  Set  for all $\omega \in \bar{ \Omega}$,
\begin{eqnarray*}
\iota(\omega) & := &  \mathrm{ess.}\inf_{n\in\bar{\mathbb{N}}}\mathfrak{U}_{n}(\omega,\infty) -1 \in \mathbb{R}\\
\mathfrak{V}_{n}(\omega,x) & := & \mathfrak{U}_{n}(\omega,x) - \iota(\omega), \, x\in \mathbb{R}, \;n\in\bar{\mathbb{N}}.
\end{eqnarray*}
Then, it is clear that \eqref{optil} for $\mathfrak{U}_{n}$ and $\mathfrak{V}_{n}$ have the same optimizers. 
Moreover, $\mathfrak{V}_{n}$ statisfies \eqref{fini},  \eqref{egyutt}, \eqref{zero} and \eqref{conv}. 
So, our conditions for $\mathfrak{V}_{n}$  also imply Assumption 2.2 and those
in Remark 2.5 of \cite{cr2007b}.  It remains to show that Assumption 2.3 of \cite{cr2007b} is also true with,
say, $\gamma=1/2$, that is, there exists $\tilde{x}>0$, such that for all $x\geq \tilde{x},$ $n\in\bar{\mathbb{N}}$
\begin{equation}\label{ae}
\mathfrak{V}_{n}(\cdot, \lambda x)\leq \lambda^{1/2}\mathfrak{V}_{n}(\cdot,x) \mbox{ a.s. }
\end{equation}
We remark that one could verify that assumption for arbitrary $\gamma>0$. \eqref{ae}
is a condition on \emph{asymptotic elasticity}, see Section 6 of \cite{ks} for
a detailed discussion of this notion.

Fix $n\in\bar{\mathbb{N}}$. To show \eqref{ae}, notice first that for all $y>0$, all $\omega \in \bar{ \Omega}$, 
\begin{eqnarray*}
\frac{y}{2}\mathfrak{V}_{n}'(\omega,y)\leq \int_{y/2}^{y}\mathfrak{V}_{n}'(\omega,t)\, dt\leq \mathfrak{V}_{n}(\omega,\infty)-\mathfrak{V}_{n}(\omega,y/2).	
\end{eqnarray*}
Choose $\tilde{y}$ so large that 
$\sup_{n\in\bar{\mathbb{N}},\omega\in\Omega}[\mathfrak{V}_{n}(\omega,\infty)-\mathfrak{V}_{n}(\omega,\tilde{y})]\leq 1/2$, this is
possible by \eqref{egyutt}. Then, since $\mathfrak{V}_{n}(\omega,\infty)\geq 1$, we have
$\mathfrak{V}_{n}(\omega,y)\geq 1/2$ for $y\geq \tilde{y}$, so
$$
0 \leq \frac{y \mathfrak{V}_{n}'(\omega,y)}{\mathfrak{V}_{n}(\omega,y)}\leq 
4[\mathfrak{V}_{n}(\omega,\infty)-\mathfrak{V}_{n}(\omega,y/2)] \leq 4 
\sup_{n\in\bar{\mathbb{N}},\omega\in\Omega}[\mathfrak{V}_{n}(\omega,\infty)-\mathfrak{V}_{n}(\omega,y/2)].$$
Hence $\frac{y \mathfrak{V}_{n}'(\omega,y)}{\mathfrak{V}_{n}(\omega,y)}$ tends to $0$
as $y\to\infty$, uniformly in $n \in\bar{\mathbb{N}}$ and $\omega\in\Omega$, by \eqref{egyutt}. 
We obtain that there is $\bar{y}\geq \tilde{y}$
such that for $y\geq \bar{y}$, $n \in\bar{\mathbb{N}}$
$$y\mathfrak{V}_{n}'(\omega ,y)<\frac12\mathfrak{V}_{ n}(\omega ,y).$$ 
Now applying the argument of $(ii)\Rightarrow (i)$ in Lemma 6.3 of \cite{ks} with the choice $\gamma=1/2$,{}
it follows that for some $\tilde{x}\geq \bar{y}$, \eqref{ae} holds.   
Now Theorem 2.1 and Remark 2.5 of \cite{cr2007b} imply the statements of our theorem.	
\end{proof}

An immediate corollary of the Ascoli theorem is noted next.

\begin{proposition}\label{arzela-ascoli}
Let $\mathcal{S}$ be a compact subset in a Euclidean space $\mathbb{R}^{N}$ with norm $|\cdot|$. Let $\Psi\subset C(\mathcal{S})$
be such that, 
$$
\sup_{\psi\in\Psi}|\psi(x)|<\infty, \, x\in\mathcal{S}
$$ 	
and, for some $\theta,A>0$, 
\begin{equation}\label{equi}
\sup_{\psi\in\Psi}|\psi(x)-\psi(y)|\leq A|x-y|^{\theta},\ x,y\in\mathcal{S}.
\end{equation}
Then $\Psi$ is relatively compact in the Banach space $C(\mathcal{S})$.
\end{proposition}
\begin{proof}
By Theorem A5 and its corollary in \cite{rudin} (see also Theorem A4), we only need to check equicontinuity of the elements of $\Psi$, which is trivial from \eqref{equi}.	
\end{proof}

Finally, we recall the celebrated theorem of Schauder, see \cite[Theorem 5.28]{rudin}.

\begin{theorem}\label{schauder} Let $\mathbb{B}$ be a Banach space, 
$H\subset \mathbb{B}$ a nonempty, compact, convex subset. If $\upsilon:H\to H$ is a continuous mapping
then there is $p\in H$ with $\upsilon(p)=p$.\hfill $\square$
\end{theorem}

\end{document}